\documentclass[]{interact}

\usepackage{amsmath}
\usepackage{amssymb}  
\usepackage{amsfonts}
\usepackage{mathtools}
\usepackage{bbm}
\usepackage{bm}
\usepackage{color}
\usepackage{graphics} 
\usepackage{graphicx}
\usepackage{epsfig} 
\usepackage{epstopdf}
\usepackage[noadjust]{cite}
\usepackage{tikz}
\usetikzlibrary{calc,positioning,shapes,shadows,arrows,fit}
\usepackage{fancyhdr}
\usepackage{fancyref}
\usepackage{hyperref}
\usepackage{framed}
\usepackage{float}
\usepackage{comment}
\usepackage[font = small]{caption}
\usepackage{subcaption}

\usepackage{array}
\newcolumntype{C}[1]{>{\centering\let\newline\\\arraybackslash}m{#1}}

\newtheorem{theorem}{Theorem}
\newtheorem{lemma}{Lemma}

\newtheorem{proposition}[theorem]{Proposition}

\theoremstyle{definition}
\newtheorem{definition}{Definition}
\newtheorem{example}{Example}
\newtheorem{assumption}{Assumption}

\theoremstyle{remark}
\newtheorem{remark}{Remark}

\theoremstyle{problem}
\newtheorem{problem}{Problem}

\newcommand{\T}{\mathcal{T}} %transition system
\newcommand{\A}{\mathcal{A}} %automaton
\newcommand{\init}{\mathit{init}}
\newcommand{\Nat}{\mathbb{N}} %natural numbers
\newcommand{\AP}{\mathit{AP}} %atomic propositions

\begin{document}

\title{On the Timed Temporal Logic Planning of Coupled Multi-Agent Systems}

\author{
	\name{Alexandros Nikou\textsuperscript{a}, Dimitris Boskos\textsuperscript{a}, Jana Tumova\textsuperscript{b} and Dimos V. Dimarogonas\textsuperscript{a}\thanks{CONTACT Alexandros Nikou. Email: anikou@kth.se, Dimitris Boskos. Email: boskos@kth.se, Jana Tumova. Email: tumova@kth.se and Dimos V. Dimarogonas. Email: dimos@kth.se}}
	\affil{\textsuperscript{a}The authors are with KTH Center of Autonomous Systems and ACCESS Linnaeus Center, School of Electrical
		Engineering, KTH Royal Institute of Technology, SE-100 44, Stockholm, Sweden. \\
		\textsuperscript{b}The author is with School of Computer Science and Communication, KTH Royal Institute of Technology, SE-100 44, Stockholm, Sweden.}
}

\thanks{This work was supported by the H2020 ERC Starting Grant BUCOPHSYS, the EU H2020 Research and
		Innovation Programme under GA No. 731869 (Co4Robots), the SSF COIN project, the Swedish Research Council (VR) and the Knut och Alice Wallenberg Foundation.}

\maketitle

\begin{abstract}
This paper presents a fully automated procedure for controller synthesis for multi-agent systems under coupling constraints. Each agent is modeled with dynamics consisting of two terms: the first one models the coupling constraints and the other one is an additional bounded control input. We aim to design these inputs so that each agent meets an individual high-level specification given as a Metric Interval Temporal Logic (MITL). First, a decentralized abstraction that provides a space and time discretization of the multi-agent system is designed. Second, by utilizing this abstraction and techniques from formal verification, we propose an algorithm that computes the individual runs which provably satisfy the high-level tasks. The overall approach is demonstrated in a simulation example conducted in MATLAB environment.
\end{abstract}


\section{Introduction}

Cooperative control of multi-agent systems has traditionally focused on designing distributed control laws in order to achieve global tasks such as consensus and formation control, and at the same time fulfill properties such as network connectivity and collision avoidance. Over the last few years, the field of control of multi-agent systems under high-level specifications has been gaining attention. In this work, we aim to additionally introduce specific time bounds into these tasks, in order to include specifications such as: ``Robot 1 and robot 2 should visit region $A$ and $B$ within 4 time units respectively or ``Both robots 1 and 2 should periodically survey regions $A_1$, $A_2$, $A_3$, avoid region $X$ and always keep the longest time between two consecutive visits to $A_1$ below 8 time units".

The qualitative specification language that has primarily been used to express the high-level tasks is Linear Temporal Logic (LTL) (see, e.g., \cite{muray_2010_receding}). There is a rich body of literature containing algorithms for verification and synthesis of multi-agent systems under high  level specifications (\cite{guo_2015_reconfiguration, zavlanos_2016_multi-agent_LTL, pappas_2016_implan}). 
%
%In order to deal with the computational complexity due to the multiple agents of We also refer to , where a discrete plan algorithm for prioritized  tasks and collision avoidance of a multi-robot systems has been proposed. We also note that the aforementioned works do not involve complex tasks with timed constraints. 
%
Controller synthesis under timed specifications has been considered in \cite{liu_MTL, murray_2015_stl, topcu_2015, baras_MTL_2016_new}. In \cite{liu_MTL}, the authors addressed the problem of designing high-level planners to achieve tasks for switching dynamical systems under Metric Temporal Logic (MTL) specification and in \cite{murray_2015_stl}, the authors utilized a counterexample-guided synthesis for cyber-physical systems subject to Signal Temporal Logic (STL) specifications. In \cite{topcu_2015}, an optimal control problem for continuous-time stochastic systems subject to objectives specified in MITL was studied. In \cite{baras_MTL_2016_new}, the authors focused on motion planning based on the construction of an efficient timed automaton from a given MITL specification. However, all these works are restricted to single agent planning and are not extendable to multi-agent systems in a straightforward way. The high-level coordination of multiple vehicles under timed specifications has been considered in \cite{frazzoli_MTL}, by solving an optimization problem over the tasks' execution time instances. % 

An automata-based solution for multi-agent systems was proposed in our previous work \cite{alex_2016_acc}, where Metric Interval Temporal Logic (MITL) formulas were introduced in order to synthesize controllers such that every agent fulfills an individual specification and the team of agents fulfills a global specification. Specifically, the abstraction of each agent's dynamics was considered to be given and an upper bound of the time that each agent needs to perform a transition from one region to another was assumed. Furthermore, potential coupled constraints between the agents were not taken into consideration. Motivated by this, in this work, we aim to address the aforementioned issues. We assume that the dynamics of each agent consists of two parts: the first part is a consensus type term representing the coupling between the agent and its neighbors, and the second one is an additional control input which will be exploited for high-level planning. Hereafter, we call it a free input. A decentralized abstraction procedure is provided, which leads to an individual Transition System (TS) for each agent and provides a basis for high-level planning. Additionally, this abstraction is associated to a time quantization which allows us to assign precise time durations to the transitions of each agent.
%
%\vspace{-2mm}
%
Abstractions for both single and multi-agent systems can be found in \cite{alur_2000_discrete_abstractions, zamani_2012_symbolic, abate_2014_finite_abstractions, PJ_tac, tabuada_compositional_abstractions, boskos_cdc_2015}. Compositional frameworks are provided in \cite{PJ_tac} for safety specifications of discrete time systems, and \cite{tabuada_compositional_abstractions}, which is focused on feedback linearizable systems with a cascade interconnection. Therefore, these results are not applicable to the systems we consider, which evolve in continuous time and do not require a specific network interconnection.

Motivated by our previous work \cite{boskos_cdc_2015}, we start from the consensus dynamics of each agent and we construct a Weighted Transition System (WTS) for each agent in a decentralized manner. Each agent is assigned an individual task given in MITL formulas. We aim to design the free inputs so that each agent performs the desired individual task within specific time bounds. In particular, we provide an automatic controller synthesis method for coupled multi-agent systems under high-level tasks with timed constraints. A motivation for this framework comes from applications such as the deployment of aerial robotic teams. In particular, the consensus coupling allows the robots to stay sufficiently close to each other and maintain a connected network during the evolution of the system. Additionally, individual MITL formulas are leveraged to assign area monitoring tasks to each robot individually. The MITL formalism enables us to impose time constraints on the monitoring process.
%The reason for the decentralization is motivated by the fact that we do not aim to give one huge task to all the robots but split the tasks, instead.
%
Compared to existing works on multi-agent planning under temporal logic specifications, the proposed approach considers dynamically coupled multi-agent systems under timed temporal specifications in a distributed way. To the best of the authors' knowledge, this is the first time that a fully automated framework for multi-agent systems consisting of both constructing an abstraction and conducting high-level timed temporal logic planning is considered.

This paper is organized as follows. In Section \ref{sec: preliminaries} a description of the necessary mathematical tools, the notations and the definitions are given. Section \ref{sec: prob_formulation} provides the dynamics of the system and the formal problem statement. Section \ref{sec: solution} discusses the technical details of the solution. Section \ref{sec: simulation_results} is devoted to a simulation example. Finally, conclusions and future work are discussed in Section \ref{sec: conclusions}.

\section{Notation and Preliminaries} \label{sec: preliminaries}

\subsection{Notation}

Denote by $\mathbb{R}, \mathbb{Q}_+, \mathbb{N}$ the set of real, nonnegative rational and natural numbers including 0, respectively. Define also $\mathbb{T}_{\infty} = \mathbb{T} \cup \{\infty\}$ for a set $\mathbb{T} \subseteq \mathbb{R}$; $\mathbb{R}_{\ge 0}$ is the set of real numbers with all elements nonnegative. Given a set $S$, we denote by $|S|$ its cardinality, by $S^N = S \times \dots \times S$, its $N$-fold Cartesian product and by $2^S$ the set of all its subsets. For a subset $S$ of $\mathbb{R}^n$, denote by $\text{cl}(S), \text{int}(S)$ and $\partial S = \text{cl}(S) \backslash \text{int}(S)$ its closure, interior and boundary, respectively, where $\backslash$ is used for set subtraction. The notation $\|x\|$ is used for the Euclidean norm of a vector $x \in \mathbb{R}^n$ and $\|A\| = \text{max} \{\|Ax\| : \|x\| = 1\}$ for the induced norm of a matrix $A \in \mathbb{R}^{m \times n}$. Given a matrix $A$ denote by $\lambda_{\text{max}}(A) = \text{max} \{|\lambda| : \lambda \in \sigma(A) \}$ the spectral radius of $A$, where $\sigma(A)$ is the set of all the eigenvalues of $A$; $A \otimes B$ denotes the Kronecker product of the matrices $A, B \in \mathbb{R}^{m \times n}$ (see \cite{horn1994topics}). Define also by $I_n \in \mathbb{R}^{n \times n}$ the identity matrix.

\subsection{Multi-Agent Systems}
\label{sec:prelims:system}

An \textit{undirected graph} $\mathcal{G}$ is a pair $(\mathcal{I}, \mathcal{E})$, where $\mathcal{I} = \{1,\dots,N\}$ is a finite set of nodes, representing a team of agents, and $\mathcal{E} \subseteq \{\{i,j\}: i,j \in \mathcal{I}, i \neq j\}$, is the set of edges that model the communication capability between the neighboring agents. For each agent, its neighbors' set $\mathcal{N}(i)$ is defined as $\mathcal{N}(i) = \{j_1, \ldots, j_{N_i} \} = \{ j \in \mathcal{I} : \{i,j\} \in \mathcal{E}\}$ where $N_i = |\mathcal N(i)|$. The Laplacian matrix $L(\mathcal{G}) \in \mathbb{R}^{N \times N}$ of the graph $\mathcal{G}$ is defined as $L(\mathcal{G}) = D(\mathcal{G}) D(\mathcal{G})^{\top}$ where $D(\mathcal{G})$ is the $N \times |\mathcal{E}|$ incidence matrix, as it is defined in \cite[Chapter 2]{mesbahi2010graph}. The graph Laplacian $L(\mathcal{G})$ is positive semidefinite and symmetric. If we consider an ordering $0 = \lambda_1(\mathcal{G}) \leq \lambda_2(\mathcal{G}) \leq \ldots \leq \lambda_N(\mathcal{G}) = \lambda_{\text{max}}(\mathcal{G})$ of the eigenvalues of $L(\mathcal{G})$ then we have that $\lambda_2(\mathcal{G}) > 0$ iff $\mathcal{G}$ is connected (\cite[Chapter 2]{mesbahi2010graph}).

Given a vector $x_i = (x^1_i, \ldots, x^n_i) \in \mathbb{R}^n$, the component operator $c(x_i, \ell) = x_i^\ell \in \mathbb{R}, \ell = 1, \ldots, n$ gives the projection of $x_i$ onto its $\ell$-th component (see \cite[Chapter 7]{mesbahi2010graph}). Similarly, for the stack vector $x = (x_1, \ldots, x_N) \in \mathbb{R}^{Nn}$ the component operator is defined as $c(x, \ell) = (c(x_1, \ell), \ldots, c(x_N, \ell)) \in \mathbb{R}^{N}, \ell = 1, \ldots, n$. By using the component operator, the norm of a vector $x \in \mathbb{R}^{Nn}$ can be evaluated as $\|x\| = \left\{ \displaystyle \sum_{\ell = 1}^{n} \| c(x, \ell) \|^2 \right\}^{\frac{1}{2}}$.

Denote by $\widetilde{x} \in \mathbb{R}^{|\mathcal{E}|n}$ the stack column vector of the vectors $x_i-x_j, \{i, j\} \in \mathcal{E}$ with the edges ordered as in the case of the incidence matrix $D(\mathcal{G})$. Then, the following holds:
\begin{equation} \label{eq:xtilde}
\widetilde{x} = \left(D(\mathcal{G})^{\top} \otimes I_n \right) x.
\end{equation}

\subsection{Cell Decompositions}

In the subsequent analysis a discrete partition of the workspace into cells will be considered which is formalized through the following definition. 

\begin{definition} \label{def: cell_decomposition}
	A \textit{cell decomposition} $S = \{S_\ell\}_{\ell \in \mathbb{I}}$ of a set $\mathcal{D} \subseteq \mathbb{R}^n$, where $\mathbb{I} \subseteq \mathbb{N}$ is a finite or countable index set, is a family of uniformly bounded convex sets $S_\ell, \ell \in \mathbb{I}$ such that $\text{int}(S_\ell) \cap \text{int}(S_{\hat{\ell}}) = \emptyset$ for all $\ell, \hat{\ell} \in \mathbb{I}$ with $\ell \neq \hat{\ell}$ and $\cup_{\ell \in \mathbb{I}} S_\ell = \mathcal D$ We assume that the interiors of the cells are non-empty.
\end{definition}

\subsection{Time Sequence, Timed Run and Weighted Transition System}

In this section we include some definitions from computer science that are required to analyze our framework. 

An infinite sequence of elements of a set $X$ is called an \textit{infinite word} over this set and it is denoted by $\chi = \chi(0)\chi(1) \ldots$. The $i$-th element of a sequence is denoted by $\chi(i)$. For certain technical reasons that will be clarified in the sequel, we will assume hereafter that $\mathbb{T} = \mathbb{Q}_+$.

\begin{definition} (\cite{alur1994}) A \textit{time sequence} $\tau = \tau(0) \tau(1) \ldots$ is an infinite sequence of time values $\tau(j) \in \mathbb{T}$, satisfying the following properties:
	\begin{itemize}
		\item Monotonicity: %$\tau$ increases strictly monotonically, i.e.,
		$\tau(j) < \tau(j+1)$ for all $j \geq 0$.
		\item Progress: For every $t \in \mathbb{T}$, there exists 	$\ j \ge 1$, such that $\tau(j) > t$.
	\end{itemize}
\end{definition}

An \textit{atomic proposition} $p$ is a statement %over the problem variables and parameters 
that is either True $(\top)$ or False $(\bot)$. % at a given time instance.

\begin{definition} (\cite{alur1994})
	Let $\AP$ be a finite set of atomic propositions. A \textit{timed word} $w$ over the set $\AP$ is an infinite sequence $w^t = (w(0), \tau(0)) (w(1), \tau(1)) \ldots$ where $w(0) w(1) \ldots$ is an infinite word over the set $2^{\AP}$ and $\tau(0) \tau(1) \ldots$ is a time sequence with $\tau(j) \in \mathbb{T}, \ j \geq 0$.
\end{definition}

\begin{definition} \label{def: WTS}
	A Weighted Transition System (\textit{WTS}) is a tuple $(S, S_0, Act, \longrightarrow, d, AP, L)$ where $S$ is a finite set of states;
	%\item
	$S_0 \subseteq S$ is a set of initial states;
	$Act$ is a set of actions;
	%\item
	$\longrightarrow \subseteq S \times Act \times S$ is a transition relation;
	$d: \longrightarrow \rightarrow \mathbb{T}$ is a map that assigns a positive weight to each transition;
	%\item
	$\AP$ is a finite set of atomic propositions; and
	%\item
	$L: S \rightarrow 2^{AP}$ is a labeling function.
	%\end{itemize}
	The notation $s \overset{\alpha}{\longrightarrow} s'$ is used to denote that $(s, \alpha, s') \in \longrightarrow$ for $s, s' \in S$ and $\alpha \in Act$. For every $s \in S$ and $\alpha \in Act$ define $\text{Post}(s, \alpha) = \{s' \in S : (s, \alpha, s') \in \longrightarrow\}$.
\end{definition}

\begin{definition}\label{run_of_WTS}
	A \textit{timed run} of a WTS is an infinite sequence $r^t = (r(0), \tau(0))(r(1), \tau(1)) \ldots$,
	such that $r(0) \in S_0$, and for all $j \geq 1$, it holds that $r(j) \in S$ and $(r(j), \alpha(j), r(j+1)) \in \longrightarrow$ for a sequence of actions $\alpha(1) \alpha(2) \ldots$ with $\alpha(j) \in Act, \forall \ j \geq 1$. The \textit{time stamps} $\tau(j), j \geq 0$ are inductively defined as:
	\begin{enumerate}
		\item $\tau(0) = 0$.
		\item $\displaystyle \tau(j+1) =  \tau(j) + d(r(j), \alpha(j), r(j+1)), \ \forall \ j \geq 1.$
	\end{enumerate}
	Every timed run $r^t$ generates a \textit{timed word}
	%\begin{align}
	$w(r^t) = %(w(0), \tau(0)) \ (w(1), \tau(1)) \ \cdots =
	(w(0), \tau(0)) \ (w(1), \tau(1)) \ldots$
	over the set $2^{\AP}$ where $w(j) = L(r(j))$, $\forall \ j \geq 0$ is the subset of atomic propositions that are true at state $r(j)$. 
	% at time $\tau(j)$. 
	
	%A word consists of the sets of Atomic Propositions valid at each location $r(j)$ and time $\tau(j)$ for $j \geq 0$.
\end{definition}

\subsection{Metric Interval Temporal Logic}

The syntax of \textit{Metric Interval Temporal Logic (MITL)} over a set of atomic propositions $AP$ is defined by the grammar:
\begin{equation*}
\varphi := p \ | \ \neg \varphi \ | \ \varphi_1 \wedge \varphi_2 \ | \ \bigcirc_I \varphi  \ | \ \Diamond_I \varphi \mid \square_I \varphi \mid  \varphi_1 \ \mathcal{U}_I \ \varphi_2,
\end{equation*}
where $p \in \AP$, and $\bigcirc$, $\Diamond$, $\square$ and $\mathcal U$ are the next, eventually, always and until temporal operator, respectively; $I = [a, b] \subseteq \mathbb{T}$ where $a, b \in [0, \infty]$ with $a < b$ is a non-empty timed interval. MITL can be interpreted either in continuous or point-wise semantics \cite{pavithra_expressiveness}. In this paper, the latter approach is utilized, since the consideration of point-wise (event-based) semantics is more suitable for the automata-based specifications considered in a discretized state-space. The MITL formulas are interpreted over timed words like the ones produced by a WTS which is given in Def. \ref{run_of_WTS}.

\begin{definition} \label{def:mitl_semantics} (\cite{pavithra_expressiveness}, \cite{quaknine_decidability})
	Given a timed word $w^t = (w(0),\tau(0))(w(1),\tau(1)) \dots$, an MITL formula $\varphi$ and a position $i$ in the timed word, the satisfaction relation $(w^t, i) \models \varphi$, for $\ i \geq 0$ (read $w^t$ satisfies $\varphi$ at position $i$) is inductively defined as follows:
	\begin{align*} \label{eq: for1}
	&(w^t, i) \models p \Leftrightarrow p \in w(i), \\
	&(w^t, i) \models \neg \varphi \Leftrightarrow (w^t, i) \not \models \varphi, \\
	&(w^t, i) \models \varphi_1 \wedge \varphi_2 \Leftrightarrow (w^t, i) \models \varphi_1 \ \text{and} \ (w^t, i) \models \varphi_2, \\
	&(w^t, i) \models \bigcirc_I \ \varphi \Leftrightarrow (w^t, i+1) \models \varphi \ \text{and} \ \tau(i+1) - \tau(i) \in I, \\
	&(w^t, i) \models \Diamond_I \varphi \Leftrightarrow \exists j \ge i, \ \text{such that} \ (w^t, j) \models \varphi, \tau(j)-\tau(i) \in {I}, \\
	&(w^t, i) \models \square_I \varphi \Leftrightarrow \forall j \ge i, \ \tau(j)-\tau(i) \in {I} \Rightarrow (w^t, j) \models \varphi,  \\
	&(w^t, i) \models \varphi_1 \ \mathcal{U}_I \ \varphi_2 \Leftrightarrow \exists j\ge i, \ \text{s.t. } (w^t, j) \models \varphi_2, \tau(j)-\tau(i) \in I \notag \\
	&\hspace{65mm} \text{and} \ (w^t, k) \models \varphi_1, \forall \ i \leq k < j.
	\end{align*}
	%We say that $r^t$ satisfies the formula $\varphi$, denoted by $r^t \models \varphi$, if $(r^t, 0) \models \varphi$. The semantics of the operators $\Diamond$ (eventually) and $\square$ (always), that are known from the LTL specifications (see \cite{katoen}), can be correspondingly defined as
	We say that a timed run $r^t = (r(0),\tau(0))(r(1),\tau(1)) \dots$ satisfies the MITL formula $\varphi$ (we write $r^t \models \varphi$) if and only if the corresponding timed word $w(r^t) = (w(0),\tau(0))(w(1),\tau(1)) \dots$ with $w(j) = L(r(j)), \forall j \ge 0$, satisfies the MITL formula ($w(r^t) \models \varphi$).
	
\end{definition}
%The next operator ($\bigcirc$) (\cite{quaknine_decidability}) from LTL formulas can be extended here as
%\begin{equation}
%\bigcirc_I \varphi = \bot \ \mathcal{U}_I \varphi
%\end{equation}

It has been proved that MITL is decidable in infinite words and point-wise semantics, which is the case considered here (see \cite{alur_mitl, reynold} for details). The model checking and satisfiability problems are \textit{EXPSPACE}-complete. It should be noted that in the context of timed systems, EXSPACE complexity is fairly low \cite{bouyer_phd}.

\begin{example}
	Consider the WTS with $S = \{s_0, s_1, s_2\}$, $S_0 = \{s_0\}$, $Act = \emptyset$,  $\longrightarrow = \{(s_0, \emptyset, s_1)$, $(s_1, \emptyset, s_2)$, $(s_1, \emptyset, s_0)$, $(s_2, \emptyset, s_1)\}$, $d((s_0, \emptyset, s_1)) = 1.0$, $d((s_1, \emptyset, s_2)) = 1.5$, $d((s_1, \emptyset, s_0)) = 2.0$, $d((s_2, \emptyset, s_1)) = 0.5$, $AP = \{green\}$, $L(s_0) = \{green\}$, $L(s_1) = L(s_2) = \emptyset$ depicted in Figure 1.
	\begin{figure}[t!] \label{fig:wts_example}
		\centering
		\begin{tikzpicture}[scale = 1.0]       
		\node(pseudo1) at (-1.2,0){};
		\node(0) [line width = 1.0] at (0,0)[shape=circle,draw][fill=green!20]          {$s_0$};
		\node(1) [line width = 1.0] at (2.5,0)[shape=circle,draw]         {$s_1$};
		\node(5) [line width = 1.0] at (5.0,0)[shape=circle,draw]        {$s_2$};
		
		\path [->] [line width = 1.0]
		(0)     edge       [bend left = 15]              node  [above]  {$1.0$}  (1)
		(1)     edge     [bend right = -15]           node  [below]  {$2.0$}     (0)
		(1)     edge     [bend right = -15]           node  [above]  {$1.5$}     (5)
		(5)     edge     [bend right = -15]           node  [below]  {$0.5$}     (1)

		% (0)      edge [loop above]    node [above]  {$0.5$}     ()
		(pseudo1) edge                                       (0);
		\end{tikzpicture}
		\caption{An example of a WTS}
	\end{figure}
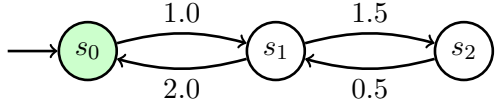
	
	\noindent Let two timed runs of the system be:
	\begin{align*}
	r_1^t &= (s_0, 0.0)(s_1, 1.0)(s_0, 3.0)(s_1, 4.0)(s_0, 6.0) \ldots, \\
	r_2^t &= (s_0, 0.0)(s_1, 1.0)(s_2, 2.5)(s_1, 3.0)(s_0, 5.0) \ldots,
	\end{align*}
	and two MITL formulas $\varphi_1 = \Diamond_{[2, 5]} \{green\}, \varphi_2 = \square_{[0, 5]} \{green\}$. According to the MITL semantics, it follows that the timed run $r_1^t$ satisfies $\varphi_1$ ($r_1^t \models \varphi_1$), since at the time stamp $3.0 \in [2, 5]$ we have that $L(s_0) = \{green\}$ so the atomic proposition $green$ occurs at least once in the given interval. On the other hand, the timed run $r_2^t$ does not satisfy $\varphi_2$ ($r_2^t \not \models \varphi_2$) since the atomic proposition $green$ does not hold at every time stamp of the run $r_2^t$ (it holds only at the time stamp $0.0$).
\end{example}

\subsection{Timed B\"uchi Automata} \label{sec: timed_automata}
\textit{Timed B\"uchi Automata (TBA)} were originally introduced in \cite{alur1994}. In this work, we partially adopt the notation from \cite{bouyer_phd, tripakis_tba}. %{is this right? did they have a buchi condition already?}
Let $C = \{c_1, \ldots, c_{|C|}\}$ be a finite set of \textit{clocks}. The set of \textit{clock constraints} $\Phi(C)$ is defined by the grammar:
\begin{equation*}
\phi :=  \top \mid \ \neg \phi \ | \ \phi_1 \wedge \phi_2 \ | \ c \bowtie \psi,
\end{equation*}
where $c \in C$ is a clock, $\psi \in \mathbb{T}$ is a clock constant and $\bowtie \ \in  \{ <, >, \geq, \leq, = \}$. A clock \textit{valuation} is a function $\nu: C \rightarrow\mathbb{T}$ that assigns a value to each clock. A clock $c_i$ has valuation $\nu_i$ for $i \in \{1, \ldots, |C|\}$, and $\nu = (\nu_1, \ldots, \nu_{|C|})$. We denote by $\nu \models \phi$ the fact that the valuation $\nu$ satisfies the clock constraint $\phi$.

%\jana{The following is connected to transitions, the above is static...}
%Let $R \subseteq X$ be a reset set. Then for every $i \in \{1, \ldots, M \}$ we have the clock valuation
%\begin{equation}
%\nu_i =
%\begin{cases}
%0 & \quad \text{if} \ x_i \in R \\
%\text{unchanged} & \quad \text{otherwise}
%\end{cases}
%\end{equation}

\begin{definition}
	A \textit{Timed B\"uchi Automaton} is a tuple $\mathcal{A} = (Q, Q^{\text{init}}, C, Inv,
	E, F, AP, \mathcal{L})$ where
	%\begin{itemize}
	%\item
	$Q$ is a finite set of locations;
	%\item
	$Q^{\text{init}} \subseteq Q$ is the set of initial locations;
	%item
	$C$ is a finite set of clocks;
	%\item
	$Inv: Q \rightarrow \Phi(C)$ is the invariant; % is invariant function a mapping that labels each location $s \in S$ with some clock constraints in $\Phi(X)$ called the invariant of the state $q$.
	%\item
	%\item
	$E \subseteq Q \times \Phi(C) \times 2^C \times Q$ gives the set of edges of the form $e = (q, \gamma, R, q')$, where $q$, $q'$ are the source and target states, $\gamma$ is the guard of edge $e$ and $R$ is a set of clocks to be reset upon executing the edge;
	%. An edge $(s, \gamma, R, s')$ represents a transition from state $s$ to state $s'$. $\gamma$ is a clock constraint over $X$ that specifies when the switch is enabled (guard). This can be a constraint either in the invariant or in the edge. The set $R \subseteq X$ represents the Reset set i.e. if a $x_i \in R$ then $\nu_i = 0, i \in \{1,\ldots, M\}$ and if $x_i \notin R$,  $\nu_i$ remains unchanged.
	%\item
	$F \subseteq Q$ is a set of accepting locations;
	$\AP$ is a finite set of atomic propositions; and
	%\item
	$\mathcal{L}: Q \rightarrow 2^{AP}$ labels every state with a subset of atomic propositions.
	%\end{itemize}
\end{definition}

A state of $\mathcal{A}$ is a pair $(q, \nu)$ where $q \in Q$ and $\nu$ satisfies the \textit{invariant} $Inv(q)$, i.e., $\nu \models Inv(q)$. The initial state of $\mathcal{A}$ is $(q(0), (0,\ldots,0))$, where $q(0) \in Q^{\text{init}}$. Given two states $(q, \nu)$ and $(q', \nu')$ and an edge $e = (q, \gamma, R, q')$, there exists a \textit{discrete transition} $(q, \nu) \overset{e}{\longrightarrow} (q', \nu')$ iff $\nu \models \gamma$, $\nu' \models Inv(q')$, and $R$ is the \textit{reset set}, i.e., $\nu'_i = 0$ for $c_i \in R$ and $\nu'_i = \nu_i$ for $c_i \notin R$. Given a $\delta \in \mathbb{T}$, there exists a \textit{time transition} $(q, \nu) \overset{\delta}{\longrightarrow} (q', \nu')$ iff $q = q', \nu' = \nu+\delta$ ($\delta$ is summed component-wise) and $\nu' \models Inv(q)$. We write $(q, \nu) \overset{\delta}{\longrightarrow} \overset{e}{\longrightarrow} (q', \nu')$ if there exists $q'', \nu''$ such that $(q, \nu) \overset{\delta}{\longrightarrow} (q'', \nu'')$ and $(q'', \nu'') \overset{e}{\longrightarrow} (q', \nu')$ with $q'' = q$.

An infinite run of $\mathcal{A}$ starting at state $(q(0), \nu)$ is an infinite sequence of time and discrete transitions $(q(0), \nu(0))\overset{\delta_0}{\longrightarrow} (q(0)', \nu(0)') \overset{e_0}{\longrightarrow} (q(1), \nu(1)) \overset{\delta_1}{\longrightarrow} (q(1)', \nu(1)') \ldots$, where $(q(0),\nu(0))$ is an initial state. %, and for all $i \geq 0$, ($s(i), \nu_i$) are states, $\delta_i \in \mathbb{T}, e_i \in E$ and $(s(i), \nu_i) \xrightarrow[~~]{\delta_i} \xrightarrow[~~]{e_i} (s(i+1), \nu_{i+1})$.

This run produces the timed word $w = (\mathcal{L}(q(0)), \tau(0))(\mathcal{L}(q(1)), \tau(1)) \ldots$ with $\tau(0) = 0$ and $\tau(i+1) = \tau(i) +\delta_i$,  $\forall \ i \geq 1$. The run is called \textit{accepting} if $q(i) \in F$ for infinitely many times. A timed word is \textit{accepted} if there exists an accepting run that produces it. The problem of deciding the language emptiness of a given TBA is PSPACE-complete \cite{alur1994}. In other words, an accepting run of a given TBA can be synthesized, if one exists. In other words, an accepting run of a given TBA can be synthesized, if one exists. Any MITL formula $\varphi$ over $AP$ can be algorithmically translated into a TBA with the alphabet $2^\AP$, such that the language of timed words that satisfy $\varphi$ is the language of timed words produced by the TBA (\cite{alur_mitl, maler_MITL_TA, nickovic_timed_aut, MITL_2_TA_tool}).
\begin{comment}
\begin{remark}
Traditionally, the clock constraints and the TBAs are defined with $\mathbb T = \Nat$, however, they can be extended to accommodate $\mathbb T = \mathbb Q_+ \cup \{0\}$. By multiplying all the rational numbers that are appearing in the state invariants and the edge constraints with their least common multiple, we have equivalently only natural numbers occurring to the TBA. For the sake of physical understanding of the timed properties of the framework under investigation, we will be working with $\mathbb{T} = \mathbb{Q}_{+} \cup \{0\}$.
\end{remark}
\end{comment}

\begin{example} \label{ex: example_1}
	A TBA with $Q = \{q_0, q_1, q_2\}, Q^{\text{init}} = \{q_0\}, C = \{c\}, Inv(q_0) = Inv(q_1) = Inv(q_2) = \emptyset, E = \{(q_0, \{c \leq c_2\}, \emptyset, q_0), (q_0, \{c \leq c_1 \vee c > c_2\}, c, q_2), (q_0, \{c \geq c_1 \wedge c \leq c_2\}, c, q_1), (q_1, \top, c, q_1), (q_2, \top,c, q_2)\}, F = \{q_1\}, AP = \{green\}, \mathcal{L}(q_0) = \mathcal{L}(q_2) = \emptyset, \mathcal{L}(q_1) = \{green\} $ that accepts all the timed words that satisfy the formula $\varphi_3 = \Diamond_{[c_1, c_2]} \{green\}$ is depicted in Figure \ref{fig:TBA_example}. This formula will be used as reference for the following examples and simulations.
	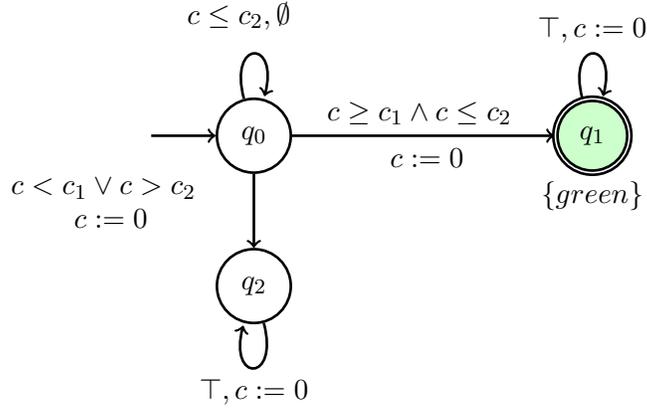
\begin{figure}[t!]
		\centering
		\begin{tikzpicture}[scale = 1.0]
		\node(pseudo) at (-1.5,0){};
		\node(0) [line width = 1.0] at (0,0)[shape=circle,draw] {$\ q_0 \ $};
		\node(1) [line width = 1.0] at (4.5,0)[shape=circle,draw, double][fill=green!20] {$\ q_1 \ $};
		\node(2) [line width = 1.0] at (0,-2)[shape=circle,draw] {$\ q_2 \ $};

		\path [->] [line width = 1.0]
		(0)      edge [bend left = 0]         node [above]  {}     (1)
		(0)      edge [loop above]    node [above]  {}     ()
		(1)      edge [loop above]    node [above]  {$\top, c := 0$}     ()
		(2)      edge [loop below]    node [below]  {$\top, c := 0$}     ()
		(0)      edge [bend left = 0]         node [above]  {}     (2)
		
		(pseudo) edge                                       (0);
		
		\node at (-0.2, 1.6) {$c \leq c_2, \emptyset$};
		\node at (2.2, 0.3) {$c \geq c_1 \wedge c \leq c_2$};
		\node at (2.3, -0.3) {$\tiny c := 0$};
		\node at (-2.0, -0.7) {$\tiny c < c_1 \vee c > c_2$};
		\node at (-1.9, -1.1) {$\tiny c := 0$};
		
		\node at (4.5,-0.8) {$\{green\}$};		
		\end{tikzpicture}
		\caption{A TBA $\mathcal{A}$ that accepts the runs that satisfy formula $\varphi = \Diamond_{[c_1, c_2]} \{green\}$.}
		\label{fig:TBA_example}
	\end{figure}
	
	An example of a timed run of this TBA is $(q_0, 0) \overset{\delta = \alpha_1}{\longrightarrow} (q_0, \alpha_1) \overset{e = (q_0, \{c \geq c_1 \wedge c \leq c_2\}, c, q_1)}{\longrightarrow} (q_1, 0) \ldots$ with $c_1 \leq \alpha_1 \leq c_2 $, which generates the timed word $w^t = (\mathcal{L}(q_0), 0)(\mathcal{L}(q_0), \alpha_1)(\mathcal{L}(q_1), \alpha_1) \ldots = (\emptyset, 0)(\emptyset, \alpha_1) (\{green\}, \alpha_1) \ldots$ that satisfies the formula $\varphi_3$. The timed run $(q_0, 0) \overset{\delta = \alpha_2}{\longrightarrow} (q_0, \alpha_2) \overset{e = (q_0, \{c \leq c_1 \vee c > c_2\}, c, q_2)}{\longrightarrow} (q_2, 0) \ldots$ with $\alpha_2 < c_1$, generates the timed word $w^t = (\mathcal{L}(q_0), 0)(\mathcal{L}(q_0), \alpha_2) (\mathcal{L}(q_2), \alpha_2) \ldots = (\emptyset, 0)(\emptyset, \alpha_2)(\emptyset, \alpha_2) \ldots$ that does not satisfy the formula $\varphi_3$.
\end{example}

\begin{remark}
	Traditionally, the clock constraints and the TBAs are defined with $\mathbb{T} = \mathbb{N}$. However, they can be extended to accommodate $\mathbb{T} = \mathbb{Q}_+$, by multiplying all the rational numbers that are appearing in the state invariants and the edge constraints with their least common multiple.
\end{remark}

\section{Problem Formulation} \label{sec: prob_formulation}

\subsection{System Model}

We focus on multi-agent systems with coupled dynamics of the form:
\begin{equation} \label{eq: system}
\dot{x}_i = -\sum_{j \in \mathcal{N}(i)}^{} (x_i - x_j)+v_{i}, x_i \in \mathbb{R}^n, i \in \mathcal{I}.
\end{equation}
The dynamics \eqref{eq: system} consists of two parts; the first part is a consensus protocol representing the coupling between the agent and its neighbors, and the second one is a control input which will be exploited for high-level planning and is called free input. In this work, it is assumed that the free inputs are bounded by a positive constant $v_{\text{max}}$, i.e., $\| v_i(t) \| \leq v_{\text{max}}, \ \forall \ i \in \mathcal{I}, t \geq 0$.

The topology of the multi-agent network is modeled through an undirected graph $\mathcal{G} = (\mathcal{I},\mathcal{E})$, where $\mathcal{I} = \{1, \ldots, N\}$ and the following assumption is made.

\begin{assumption}
	The communication graph $\mathcal{G} = (\mathcal{I},\mathcal{E})$ of the system is undirected, connected and static i.e., every agent preserves the same neighbors for all times.
\end{assumption}

\subsection{Specification}
Our goal is to control the multi-agent system \eqref{eq: system} so that each agent's behavior obeys a desired individual specification $\varphi_i$ given in MITL. In particular, it is required to drive each agent to a desired subset of the \textit{workspace} $\mathbb R^n$ within certain time limits and provide certain atomic tasks there. Atomic tasks are captured through a finite set of \textit{services} $\Sigma_i, i \in \mathcal{I}$. The position $x_i$ of each agent $i \in \mathcal{I}$ is labeled with services that are offered there. Thus, a \textit{service labeling function}:
\begin{equation} \label{eq:label_lambda}
\Lambda_i:\mathbb{R}^n\to 2^{\Sigma_i},
\end{equation}
is introduced for each agent $i \in \mathcal{I}$ which maps each state $x_i \in \mathbb{R}^n$ to the subset of services $\Lambda_i(x_i)$ which hold true at $x_i$ i.e., the subset of services that the agent $i$ \textit{can} provide in position $x_i$. It is noted that although the term service labeling function it is used, these functions are not necessarily related to the labeling functions of a WTS as in Definition \ref{def: WTS}. Define also by $\Lambda(x) = \bigcup_{i \in \mathcal I} \Lambda_i(x)$ the union of all the service labeling functions. We also assume that $\Sigma_i \cap \Sigma_j = \emptyset$, for all $i,j \in \mathcal{I}, i \neq j$ which means that the agents do not share any services. Let us now introduce the following assumption which is necessary for formally defining the problem.
\begin{assumption}  \label{assumption: AP_cell_decomposition}
	There exists a decomposition $S = \{S_\ell\}_{\ell \in \mathbb I}$ of the workspace $\mathbb{R}^n$ which forms a cell decomposition according to Def. \ref{def: cell_decomposition} and respects the labeling function $\Lambda$ i.e., for all $S_\ell \in S$ it holds that $\Lambda(x) = \Lambda(x'), \forall \ x, x' \in S_\ell$. This assumption implies that the same services hold at all the points that belong to the same cell of the decomposition.
\end{assumption}
\noindent Define for each agent $i \in \mathcal{I}$ the labeling function:
\begin{equation} \label{eq:label_mathcal_lambda}
\mathcal{L}_i: S \to 2^{\Sigma_i},
\end{equation}
which denotes the fact that when agent $i$ visits a region $S_\ell \in S$ and it can choose to \textit{provide} a subset of the services that are being offered there i.e. it chooses to satisfy a subset of $\mathcal{L}_i(S_\ell)$.

The trajectory of each agent $i$ is denoted by $x_i(t), t \geq 0, i \in \mathcal{I}$. The trajectory $x_i(t)$ is associated with a unique sequence $r_{x_i}^t = (r_i(0), \tau_i(0))(r_i(1), \tau_i(1))(r_i(2), \tau_i(2))\ldots$, of regions that the agent $i$ crosses, where for all $j \ge 0$ it holds that: $x_i(\tau_i(j)) \in r_i(j)$ and $\Lambda_i(x_i(t)) = \mathcal{L}_i(r_i(j)), \forall \ t \in [\tau_i(j), \tau_i(j+1))$ for some $r_i(j) \in S$ and $r_i(j) \ne r_i(j+1)$. %The equality $\Lambda_i(\cdot) = \mathcal{L}_i(\cdot), i \in \mathcal{I}$ is feasible due to Assumption \ref{assumption: AP_cell_decomposition}. 
The timed word $w_{x_i}^t = (\mathcal{L}_i(r_i(0)), \tau_i(0))(\mathcal{L}_i(r_i(1)), \tau_i(1))(\mathcal{L}_i(r_i(2)), \tau_i(2))\ldots$, where $w_i(j) = \mathcal{L}_i(r_i(j)), j \ge 0, i \in \mathcal{I}$, is associated uniquely with the trajectory $x_i(t)$, and represents the sequence of services that \textit{can be provided} by the agent $i$ following the trajectory $x_i(t), t \ge 0$. 

\noindent Define a \textit{timed service word} as:
\begin{align} \label{eq:time_serviced_word}
\widetilde{w}_{x_i}^t &= (\beta_i(z_0), \widetilde{\tau}_i(z_0))(\beta_i(z_1), \widetilde{\tau}_i(z_1))(\beta_i(z_2), \widetilde{\tau}_i(z_2)) \ldots,
\end{align}
where $z_0 =0 < z_1 < z_2 < \ldots$ is a sequence of integers, and for all $j \ge 0$ it holds that $\beta_i(z_j) \subseteq \mathcal{L}_i(r_i(z_j))$ and $\widetilde{\tau}_i(z_j) \in [\tau_i(z_j), \tau_i(z_j+1))$. The timed service word is a sequence of services that are actually provided by agent $i$ and is compliant with the trajectory $x_i(t), t \ge 0$ by construction.

The specification task $\varphi_i$ given as an MITL formula over the set of services $\Sigma_i$ as in Def. \ref{def:mitl_semantics}, captures requirements on the services to be provided by agent $i$, for each $i \in \mathcal{I}$. We say that a trajectory $x_i(t)$ satisfies a formula $\varphi_i$  given in MITL over the set $\Sigma_i$, and formally write $x_i(t) \models \varphi_i, \forall t \ge 0$, if and only if there exists a \textit{timed service word} $\widetilde{w}_{x_i}^t$ that complies with $x_i(t)$ and satisfies $\varphi_i$ according to the semantics of Def. \ref{def:mitl_semantics}.

\begin{example} \label{ex: example_01}
	Consider %here an example in order to visualize the notation and the technical terms that have been introduced until now. Let 
	$N = 2$ agents performing in the partitioned environment of Figure \ref{fig: example_01}. Both agents have the ability to pick up, deliver and throw two different balls. Their sets of services are $\Sigma_1 = \{\rm pickUp1, deliver1, throw1\}$ and $\Sigma_2 = \{\rm pickUp2, deliver2, throw2\}$, respectively, and satisfy $\Sigma_1 \cap \Sigma_2 = \emptyset$. Three points of the agents' trajectories that belong to different cells with different services are captured. Assume that $t_1 < t_1' < t_2 < t_2 < t_2' < t_3 < t_3'.$ The trajectories $x_1(t), x_2(t), t \ge 0$ are depicted with the red lines. According to Assumption \ref{assumption: AP_cell_decomposition}, the cell decomposition $S = \{S_\ell\}_{\ell \in \mathbb I} = \{S_1, \ldots, S_6\}$ is given where $\mathbb I = \{1, \ldots, 6\}$ respects the labeling functions $\Lambda_i, \mathcal{L}_i, i \in \{1,2\}$. In particular, it holds that: 
	\begin{align}
	&\Lambda_1(x_1(t)) = \mathcal{L}_1(r_1(0)) = \{\rm pickUp1\}, t \in [0, t_1), \notag \\
	&\Lambda_1(x_1(t)) = \mathcal{L}_1(r_1(1)) = \{\rm throw1\}, t \in [t_1, t_2), \notag \\ 
	&\Lambda_1(x_1(t)) = \mathcal{L}_1(r_1(2)) = \{\rm deliver1\}, t \in [t_2, t_3), \notag \\
	&\Lambda_1(x_1(t)) = \mathcal{L}_1(r_1(3)) = \emptyset, t \ge t_3. \notag \\
	&\Lambda_2(x_2(t)) = \mathcal{L}_2(r_2(0)) = \{\rm pickUp2\}, t \in [0, t_1'), \notag \\
	&\Lambda_2(x_2(t)) = \mathcal{L}_2(r_2(1))  = \{\rm deliver2\}, t \in [t_1', t_2'), \notag \\ 
	&\Lambda_2(x_2(t)) = \mathcal{L}_2(r_2(2)) = \{\rm throw2\}, t \in [t_2', t_3'), \notag \\
	&\Lambda_2(x_2(t)) = \mathcal{L}_2(r_2(3))  = \emptyset, t \ge t_3'. \notag
	\end{align}
	By the fact that $w_i(j) = \mathcal{L}(r_i(j)), \forall \ i \in \{1,2\}, j \in \{1,2,3\}$, the corresponding individual timed words are given as:
	\begin{align}
	w^t_{x_1} &= (\{\rm pickUp1\}, 0)(\{\rm throw1\}, t_1)(\{\rm deliver1\}, t_2)(\emptyset, t_3), \notag \\
	w^t_{x_2} &= (\{\rm pickUp2\}, 0)(\{\rm deliver2\}, t_1')(\{\rm throw2\}, t_2')(\emptyset, t_3'). \notag
	\end{align}
	According to \eqref{eq:time_serviced_word}, two time service words (depicted with in Figure \ref{fig: example_01}) are given as:
	\begin{align*}
	\widetilde{w}_1^t &= (\beta_1(z_0), \widetilde{\tau}_1(z_0))(\beta_1(z_1), \widetilde{\tau}_1(z_1)), \notag \\ 
	\widetilde{w}_2^t &= (\beta_2(z'_0), \widetilde{\tau}_2(z'_0))(\beta_2(z'_1), \widetilde{\tau}_2(z'_1)), \notag \\ 
	\end{align*}
	where for agent 1 we have: $z_0 = 0, z_1 = 2, \beta_1(z_0) = \{\rm pickUp1\} \subseteq \mathcal{L}_1(r_1(z_0)), \beta_1(z_1) = \{\rm deliver1\} \subseteq \mathcal{L}_1(r_1(z_1))$. The corresponding elements for agent 2 are $z'_0 = 0, z'_1 = 2, \beta_2(z'_0) = \{\rm pickUp2\} \subseteq \mathcal{L}_2(r_2(z'_0)), \beta_2(z'_1) = \{\rm throw2\} \subseteq \mathcal{L}_2(r_2(z'_1))$. The time stamps $\widetilde{\tau}_1(z_0), \widetilde{\tau}_1(z_1)$ should satisfy the following conditions:
	\begin{align*}
	\widetilde{\tau}_1(z_0) &\in [\tau_1(z_0), \tau_1(z_0+1)) = [0, t_1), \\
	\widetilde{\tau}_1(z_1) &\in [\tau_1(z_1), \tau_1(z_1+1)) = [t_2, t_3), \\
	\widetilde{\tau}_2(z'_0) &\in [\tau_2(z'_0), \tau_2(z'_1)) = [0, t_1'), \\ 
	\widetilde{\tau}_2(z_1') &\in [\tau_1(z_1'), \tau_1(z_1'+1)) = [t_2', t_3').
	\end{align*}
	
	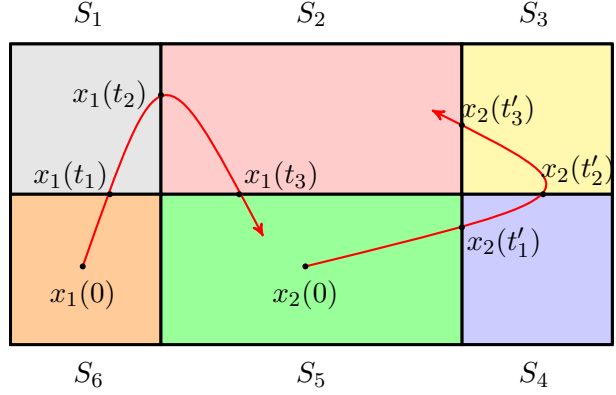
\begin{figure}[t!]
		\centering
		\begin{tikzpicture}[scale = 0.8]
		
		% plot the grid		
		\draw[step=2.5, line width=.04cm] (-2.5, -5.0) grid (0,0);
		\draw[line width=.04cm] (-7.5,0.0) -- (-2.5,0.0);
		\draw[line width=.04cm] (-7.5,-2.5) -- (-2.5,-2.5);
		\draw[line width=.04cm] (-7.5,-5.0) -- (-2.5,-5.0);
		\draw[step=2.5, line width=.04cm] (-10.0, -5.0) grid (-7.5,0);
		
		% plot the colours
		\filldraw[fill=black!10, line width=.04cm] (-10, -2.5) rectangle +(2.5, 2.5);
		\filldraw[fill=orange!40, line width=.04cm] (-10, -5.0) rectangle +(2.5, 2.5);
		\filldraw[fill=red!20, line width=.04cm] (-7.5, -2.5) rectangle +(5.0, 2.5);
		\filldraw[fill=yellow!40, line width=.04cm] (-2.5, -2.5) rectangle +(2.5, 2.5);
		\filldraw[fill=blue!20, line width=.04cm] (-2.5, -5.0) rectangle +(2.5, 2.5);
		\filldraw[fill=green!40, line width=.04cm] (-7.5, -5.0) rectangle +(5.0, 2.5);
		
		% draw the trajectory for agent 1
		\draw [color=red,thick,->,>=stealth'](-8.8, -3.7) .. controls (-7.50, -0.0) .. (-5.8, -3.2);
		
		% draw the trajectory for agent 2
		\draw [color=red,thick,->,>=stealth'](-5.1, -3.7) .. controls (-0.2, -2.5) .. (-3.0, -1.1);
		
		% draw the points of agent 1
		\draw (-8.8, -3.7) node[circle, inner sep=0.8pt, fill=black, label={below:{$x_1(0)$}}] (A1) {};
		\draw (-8.35, -2.5) node[circle, inner sep=0.8pt, fill=black] (B1) {};
		\draw (-7.50, -0.85) node[circle, inner sep=0.8pt, fill=black, label={left:{$x_1(t_2)$}}] (C1) {};
		\draw (-6.2, -2.5) node[circle, inner sep=0.8pt, fill=black] (D1) {};
		
		% draw the points of agent 1
		\draw (-1.15, -2.5) node[circle, inner sep=0.8pt, fill=black] (A2) {};
		\draw (-2.5, -3.05) node[circle, inner sep=0.8pt, fill=black] (B2) {};
		\draw (-5.1, -3.7) node[circle, inner sep=0.8pt, fill=black, label={below:{$x_2(0)$}}] (C2) {};
		\draw (-2.5, -1.35) node[circle, inner sep=0.8pt, fill=black] (D2) {};
		
		% cell decomposition
		\node at (-8.7, 0.5) {$S_1$};
		\node at (-5.0, 0.5) {$S_2$};
		\node at (-1.3, 0.5) {$S_3$};
		\node at (-1.3, -5.5) {$S_4$};
		\node at (-5.0, -5.5) {$S_5$};
		\node at (-8.7, -5.5) {$S_6$};
		\node at (-9,-2.2) {$x_1(t_1)$};
		\node at (-5.5,-2.2) {$x_1(t_3)$};
		\node at (-1.8,-3.3) {$x_2(t_1')$};
		\node at (-0.57,-2.1) {$x_2(t_2')$};
		\node at (-1.9,-1.1) {$x_2(t_3')$};
		\end{tikzpicture}
		
		\caption{An example of two agents performing in a partitioned workspace.}
		\label{fig: example_01}
	\end{figure}
\end{example}

\subsection{Problem Statement}

We are now ready to define the problem treated in this paper formally as follows:

\begin{problem} \label{problem: basic_prob}
	Given $N$ agents that are governed by dynamics as in \eqref{eq: system}, modeled by the communication graph $\mathcal{G}$, $N$ task specification formulas $\varphi_1, \ldots, \varphi_N$ expressed in MITL over the sets of services $\Sigma_1, \ldots, \Sigma_{{N}}$, respectively, service labeling functions $\Lambda_1, \ldots, \Lambda_N$, as in \eqref{eq:label_lambda}, a cell decomposition $S = \{S_\ell\}_{\ell \in \mathbb I}$ as in Assumption \ref{assumption: AP_cell_decomposition} and the labeling functions $\mathcal{L}_1, \ldots, \mathcal{L}_N$ given by \eqref{eq:label_mathcal_lambda}, assign control laws to the free inputs $v_1, \ldots, v_N$ such that each agent fulfills its individual specification i.e., $x_i(t) \models \varphi_i, \forall i \in \mathcal{I}, t \ge 0$, given the upper bound $v_{\text{max}}$.
\end{problem}

It should be noted that, in this work, the dependencies between the agents are induced through the coupled dynamics \eqref{eq: system}  and not in the discrete level, by allowing for couplings between the services (i.e., $\Sigma_i \cap \Sigma_j \ne \emptyset$, for some $i, j \in \mathcal I$). Hence, even though the agents do not share atomic propositions, the constraints on their motion due to the dynamic couplings may restrict them to fulfill the desired high-level tasks. Treating additional couplings through individual atomic propositions in the discrete level, constitutes a far from trivial problem, which is a topic of current work.

\begin{remark}
	In our previous work on the multi-agent controller synthesis framework under MITL specifications \cite{alex_2016_acc}, the multi-agent system was considered to have fully-actuated dynamics. The only constraints on the system were due to the presence of time constrained MITL formulas. In the current framework, we have two types of constraints: the constraints due to the coupling dynamics of the system \eqref{eq: system}, which constrain the motion of each agent, and, the timed constraints that are inherently imposed from the time bounds of the MITL formulas. Thus, there exist formulas that cannot be satisfied either due to the coupling constraints or the time constraints of the MITL formulas. These constraints, make the procedure of the controller synthesis in the discrete level substantially different and more elaborate than the corresponding multi-agent LTL frameworks in the literature (\cite{guo_2015_reconfiguration, frazzoli_vehicle_routing, belta_2010_product_system, belta_cdc_reduced_communication}).
\end{remark}

\begin{remark}
	The motivation for introducing the cell decomposition $S = \{S_\ell\}_{\ell \in \mathbb I}$ in this Section, comes from the requirement to know a priori which services hold in each part of the workspace. As will be clarified through the problem solution, this is necessary since the abstraction of the workspace (which is part of our proposed solution) may not be compliant with the initial given cell decomposition, and thus, new cell decompositions might be required.
\end{remark}

\section{Proposed Solution} \label{sec: solution}

In this section, a systematic solution to Problem~\ref{problem: basic_prob} is introduced. Our overall approach builds on abstracting system \eqref{eq: system} though a WTS for each agent and exploiting the fact that the timed runs in the $i$-th WTS project onto the trajectories of agent $i$ while preserving the satisfaction of the individual MITL formulas $\varphi_i, i \in \mathcal{I}$. The following analysis is performed:
\begin{enumerate}
	\item Initially, the boundedness of the agents' relative positions is proved, in order to guarantee boundedness of the coupling terms $-\sum_{j \in \mathcal{N}(i)}^{} (x_i - x_j)$. This property is required for the derivation of the symbolic models. (Section \ref{sec: boundedeness}).
	\item We utilize decentralized abstraction techniques for the multi-agent system, i.e., a discretization of both the workspace and time in order to model the motion capabilities of each agent by a WTS $\mathcal{T}_i, \ i \in \mathcal{I}$ (Section \ref{sec: abstration}).
	\item Given the WTSs, consistent runs are defined in order to take into consideration the coupling constraints among the agents. The computation of the product of the individual WTSs is also required (Section \ref{sec: runs consistency}).
	\item A five-step automated procedure for controller synthesis which serves as a solution to Problem \ref{problem: basic_prob} is provided in Section \ref{sec: synthesis}.
	\item Finally, the computational complexity of the proposed approach is discussed in Section \ref{sec:complexity}.
\end{enumerate}
The next sections provide the proposed solution in detail.

\subsection{Boundedness Analysis} \label{sec: boundedeness}

\begin{theorem} \label{theorem: theorem_1}
	Consider the multi-agent system \eqref{eq: system} modeled by the undirected communication graph $\mathcal{G}$. Assume that the network graph is connected (i.e. $\lambda_2(\mathcal{G}) > 0$) and let $v_i, i \in \mathcal{I}$ satisfy $\| v_i(t) \| \leq v_{\text{max}}, \ \forall \ i \in \mathcal{I}, t \geq 0$. Furthermore, let $\bar{R} > K_2 v_{\text{max}}$ be a positive constant, where $K_2 = \frac{2 \sqrt{N} (N-1) \left\| D(\mathcal{G})^{\top} \right\|}{\lambda_2^2(\mathcal{G})} > 0$ and where $D(\mathcal{G})$ is the network adjacency matrix. Then, for each initial condition $x_i(0) \in \mathbb{R}^n$, there exists a time $T > 0$ such that $\widetilde{x}(t) \in \mathcal{X}, \ \forall t \geq T$, where $\mathcal{X} = \{x \in \mathbb{R}^{Nn} : \|\widetilde{x}\| \leq \bar{R}\}$ and with $\widetilde{x}(t)$ as given in \eqref{eq:xtilde}.
\end{theorem}

\begin{proof}
	The proof can be found in Appendix \ref{app:proof_theorem_1}.
\end{proof}

It should be noticed that the relative boundedness of the agents' positions guarantees a global bound on the coupling terms $-\sum_{j \in \mathcal{N}(i)}^{} (x_i - x_j)$, as defined in \eqref{eq: system}. This bound will be later exploited in order to capture the behavior of the system in $\mathcal{X} = \{x \in \mathbb{R}^{Nn} : \|\widetilde{x}\| \leq \bar{R}\}$, by a discrete state WTS.

\subsection{Abstraction} \label{sec: abstration}

In this section we provide the abstraction technique that is adopted from our previous work \cite{boskos_cdc_2015} in order to capture the dynamics of each agent into Transition Systems. Thereafter, we work completely in discrete level, which is necessary in order to solve Problem \ref{problem: basic_prob}.

Firstly, some additional notation is introduced. Given an index set $\mathbb{I}$ and an agent $i \in \mathcal{I}$ with neighbors $j_1, \ldots, j_{N_i} \in \mathcal{N}(i)$, define the mapping $\text{pr}_i : \mathbb{I}^N \rightarrow \mathbb{I}^{N_i+1}$ which assigns to each $N$-tuple ${\bf{l}}= (l_1, \ldots, l_N) \in \mathbb{I}^N$ the $N_i+1$ tuple ${\bf{l}}_i = (l_i, l_{j_1}, \ldots, l_{j_{N_i}}) \in \mathbb{I}^{N_i+1}$ which denotes the indices of the cells where the agent $i$ and its neighbors belong.

\subsubsection{Well-Posed Abstractions}

Loosely speaking, an abstraction is characterized by a discretization of the workspace into cells, which we denote by $\bar{S} = \{\bar{S}_{l}\}_{l \in \bar{\mathbb{I}}}$, a time step $\delta t$ and selection of feedback laws in place of the free inputs $v_i(t), \forall i \in \mathcal{I}$. The time step $\delta t$ models the time that an agent needs to transit from one cell to another, and $v_i(t)$ is the controller that guarantees such a transition. Note that the time step $\delta t$ is the same for all the agents. Let us denote by $(\bar{S},\delta t)$ the aforementioned \emph{space-time discretization}. 

Before defining formally the concept of well-posed abstractions, an intuitive graphical representation is provided. %We next illustrate the concept of a well-posed abstraction, namely, a discretization which generates for each agent a Transition System in accordance with the discussion above and the Def. \ref{def: WTS}. 
Consider a cell decomposition $\bar{S} = \{\bar{S}_{l}\}_{l \in \bar{\mathbb{I}} = \{1, \ldots, 12\}}$ as depicted in Figure \ref{fig: well_possed_abstraction} and a time step $\delta t$. The tails and the tips of the arrows in the figure depict the initial cell and the endpoints of agent's $i$ trajectories at time $\delta t$ respectively. In both cases in the figure we focus on agent $i$ and consider the same cell configuration for $i$ and its neighbors. By configuration we mean the cell that the agent $i$ and its neighbors belong at a current time. However, different dynamics are considered for Cases (i) and (ii). In Case (i), it can be observed that for the three distinct initial positions in cell $\bar{S}_{l_i}$, it is possible to drive agent $i$ to cell $\bar{S}_{l_i'}$ at time $\delta t$. We assume that this is possible for all initial conditions in this cell and irrespectively of the initial conditions of $i$'s neighbors in their cells and the inputs they choose. It is also assumed that this property holds for all possible cell configurations of $i$ and for all the agents of the system. Thus, we have a \emph{well-posed discretization} for system (i). On the other hand, for the same cell configuration and system (ii), the following can be observed. For three distinct initial conditions of $i$ the corresponding reachable sets at $\delta t$, which are enclosed in the dashed circles, lie in different cells. Thus, it is not possible given this cell configuration of $i$ to find a cell in the decomposition which is reachable from every point in the initial cell and we conclude that discretization is not well-posed for system (ii).

\begin{figure}[t!] 
	\begin{center}
		\begin{tikzpicture} [scale=.80]
		
		\draw[color=gray,thick] (0,0) -- (4,0);
		\draw[color=gray,thick] (0,1) -- (4,1);
		\draw[color=gray,thick] (0,2) -- (4,2);
		\draw[color=gray,thick] (0,3) -- (4,3);
		
		\draw[color=gray,thick] (0,0) -- (0,3);
		\draw[color=gray,thick] (1,0) -- (1,3);
		\draw[color=gray,thick] (2,0) -- (2,3);
		\draw[color=gray,thick] (3,0) -- (3,3);
		\draw[color=gray,thick] (4,0) -- (4,3);
		
		\draw[color=blue,very thick] (0,0) -- (1,0) -- (1,1) -- (0,1) -- (0,0);
		\draw[color=green,very thick] (2,1) -- (3,1) -- (3,2) -- (2,2) -- (2,1);
		\draw[color=green,very thick] (3,2) -- (4,2) -- (4,3) -- (3,3) -- (3,2);
		\draw[color=cyan,very thick] (0,2) -- (1,2) -- (1,3) -- (0,3) -- (0,2);
		
		\fill[yellow] (0.3,2.5) circle (0.3cm);
		\fill[yellow] (0.4,2.7) circle (0.3cm);
		\fill[yellow] (0.5,2.6) circle (0.3cm);
		
		\draw[black, dashed] (0.3,2.5) circle (0.3cm);
		\draw[black, dashed] (0.4,2.7) circle (0.3cm);
		\draw[black, dashed] (0.5,2.6) circle (0.3cm);

		\draw [color=red,thick,->,>=stealth'](0.2,0.3) .. controls (0.3,1.5) .. (0.3,2.5);
		\fill[black] (0.2,0.3) circle (1.5pt);
		
		\draw [color=red,thick,->,>=stealth'](0.5,0.8) .. controls (0.5,1.5) .. (0.4,2.7);
		\fill[black] (0.5,0.8) circle (1.5pt);
		
		\draw [color=red,thick,->,>=stealth'](0.7,0.6) .. controls (0.7,1.5) .. (0.5,2.6);
		\fill[black] (0.7,0.6) circle (1.5pt);

		\coordinate [label=left:$\bar{S}_{l_i}$] (A) at (0,0.5);
		\coordinate [label=above left:$x_{i}$] (A) at (0.85,0.15);
		\coordinate [label=right:$\bar{S}_{l_{j_1}}$] (A) at (3,1.5);
		\fill[black] (2.3,1.4) circle (1.5pt) node[right]{$x_{j_{1}}$};
		\coordinate [label=right:$\bar{S}_{l_{j_2}}$] (A) at (4,2.5);
		\fill[black] (3.4,2.8) circle (1.5pt) node[below]{$x_{j_{2}}$};
		\coordinate [label=left:$\bar{S}_{l_i'}$] (A) at (0,2.5);
		
		\draw[color=gray,thick] (5,0) -- (9,0);
		\draw[color=gray,thick] (5,1) -- (9,1);
		\draw[color=gray,thick] (5,2) -- (9,2);
		\draw[color=gray,thick] (5,3) -- (9,3);
		
		\draw[color=gray,thick] (5,0) -- (5,3);
		\draw[color=gray,thick] (6,0) -- (6,3);
		\draw[color=gray,thick] (7,0) -- (7,3);
		\draw[color=gray,thick] (8,0) -- (8,3);
		\draw[color=gray,thick] (9,0) -- (9,3);
		
		\draw[color=blue,very thick] (5,0) -- (6,0) -- (6,1) -- (5,1) -- (5,0);
		\draw[color=green,very thick] (7,1) -- (8,1) -- (8,2) -- (7,2) -- (7,1);
		\draw[color=green,very thick] (8,2) -- (9,2) -- (9,3) -- (8,3) -- (8,2);
		\draw[color=cyan,dashed,very thick] (5,2) -- (6,2) -- (6,3) -- (5,3) -- (5,2);
		\draw[color=cyan,dashed,very thick] (6,2) -- (7,2) -- (7,3) -- (6,3) -- (6,2);
		\draw[color=cyan,dashed,very thick] (6,1) -- (7,1) -- (7,2) -- (6,2) -- (6,1);

		\fill[yellow] (5.3,2.5) circle (0.3cm);
		\fill[yellow] (6.4,2.7) circle (0.3cm);
		\fill[yellow] (6.5,1.7) circle (0.3cm);
		
		\draw[black, dashed] (5.3,2.5) circle (0.3cm);
		\draw[black, dashed] (6.4,2.7) circle (0.3cm);
		\draw[black, dashed] (6.5,1.7) circle (0.3cm);
		
		\draw [color=red,thick,->,>=stealth'](5.2,0.3) .. controls (5.3,1.5) .. (5.3,2.5);
		\fill[black] (5.2,0.3) circle (1.5pt);
		
		\draw [color=red,thick,->,>=stealth'](5.5,0.8) .. controls (5.5,1.5) .. (6.4,2.7);
		\fill[black] (5.5,0.8) circle (1.5pt);
		
		\draw [color=red,thick,->,>=stealth'](5.7,0.6) .. controls (5.7,1) .. (6.5,1.7);
		\fill[black] (5.7,0.6) circle (1.5pt);

		\coordinate [label=below:System (i)] (A) at (2,-0.5);
		\coordinate [label=below:System (ii)] (A) at (7,-0.5);
		
		\coordinate [label=left:$\bar{S}_{l_i}$] (A) at (5,0.5);
		\coordinate [label=above left:$x_{i}$] (A) at (5.85,0.15);
		\coordinate [label=right:$\bar{S}_{l_{j_1}}$] (A) at (8,1.5);
		\fill[black] (7.3,1.4) circle (1.5pt) node[right]{$x_{j_{1}}$};
		\coordinate [label=right:$\bar{S}_{l_{j_2}}$] (A) at (9,2.5);
		\fill[black] (8.4,2.8) circle (1.5pt) node[below]{$x_{j_{2}}$};
		
		\draw[dashed,->,>=stealth] (0.4,2.7) -- (1.5,3.5);
		\coordinate [label=right:$x_{i}(\delta t)$] (A) at (1.5,3.5);
		
		\draw[dashed,->,>=stealth] (5.3,2.5) -- (6.5,3.5);
		\coordinate [label=right:$x_{i}(\delta t)$] (A) at (6.5,3.5);
		
		\end{tikzpicture}
		\caption{Illustration of a space-time discretization which is well posed for system (i) but non-well posed for system (ii).}
		\label{fig: well_possed_abstraction}
	\end{center}
\end{figure}
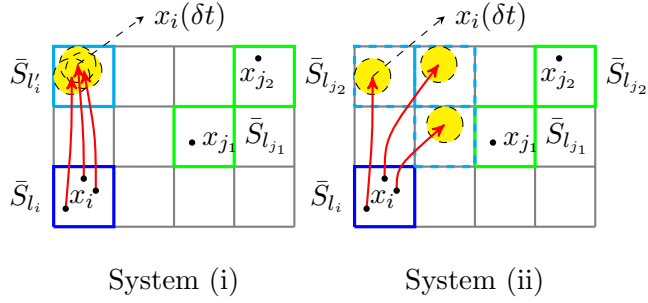

More specifically, consider a $(\bar{S},\delta t)$-space-time discretization which is the outcome of the abstraction technique that is designed for the problem solution and will be presented in Section \ref{sec:discrete_system_abstraction}. Let $\bar{S} = \{\bar{S}_{l}\}_{l \in \bar{\mathbb{I}}}$ be a cell decomposition in which the agent $i$ occupies the cell $\bar{S}_{l_i}$, $\delta t$ be a time step and $\bar{d}_{\text{max}} = \text{sup} \{ \|x - y\| : x,y \in \bar{S}_l, l \in \mathbb{I} \}$ be the diameter of the cell decomposition $\bar{S}$. It should be noted that this decomposition is not necessarily the same cell decomposition $S$ from Assumption \ref{assumption: AP_cell_decomposition} and Problem \ref{problem: basic_prob}. Through the aforementioned space and time discretization $(\bar{S},\delta t)$ we aim to capture the reachability properties of the continuous system \eqref{eq: system}, in order to create a WTS for each agent. %Let us investigate now what well-posed abstraction means for a given $(\bar{S},\delta t)$-space-time discretization. 
If there exists a free input for each state in $\bar{S}_{l_i}$ that navigates the agent $i$ into the cell $\bar{S}_{l_i'}$ precisely in time $\delta t$, regardless of the locations of the agent $i$'s neighbors within their current cells, then a transition from $l_i$ to $\l_i'$ is enabled in the WTS. This forms the well-possessedness of transitions. A more detailed mathematical derivation as well as feedback laws $v_i(t), i \in \mathcal{I}$ which guarantee a well-posed space-time discretization $(\bar{S},\delta t)$ can be found in our previous work \cite{boskos_cdc_2015}.

\subsubsection{Sufficient Conditions}

We present at this point the sufficient conditions that relate the dynamics of the multi-agent system \eqref{eq: system}, the time step $\delta t$ and the diameter $\bar{d}_{\text{max}}$, and guarantee the existence of the aforementioned well-posed transitions for each cell. Based on our previous work \cite{boskos_cdc_2015} (Section III, inequality (3), Section IV, inequalities (28, 29)), in order to derive well-posed abstractions, a nonlinear system of the form:
\begin{equation} \label{eq:boskos_dynamics}
\dot{x}_i = f_i(x_i, \mathbf{x}_j)+v_i, i \in \mathcal{I},
\end{equation}
where $\mathbf{x}_j = (x_{j_1}, \ldots, x_{j_{N_i}}) \in \mathbb{R}^{N_i n}$, should fulfill the following sufficient conditions:

\noindent \textbf{(C1)} There exists $M > v_{\text{max}} > 0$ such that $\|f_i(x_i, \mathbf{x}_j) \| \leq M, \ \forall i \in \mathcal{I}, \forall \ x \in \mathbb{R}^{Nn} : \text{pr}_i(x) = (x_i, \mathbf{x}_j) \ \text{and} \  \widetilde{x} \in \mathcal{X}$, by applying the projection operator $\text{pr}_i$ for $\mathbb{I} = \mathbb{R}^n$.

\noindent \textbf{(C2)} There exists a Lipschitz constant $L_1 > 0$ such that:
\begin{align}
\| f_i(x_i, \mathbf{x}_j) - f_i(x_i, \mathbf{y}_j) \| \leq L_1 \|(x_i, \mathbf{x}_j) - (x_i, \mathbf{y}_j) \|, \forall \ i \in \mathcal{I}, x_i, y_i \in \mathbb{R}^n, \mathbf{x}_j, \mathbf{y}_j \in \mathbb{R}^{N_in}. 
\label{eq:lipsitch_1}
\end{align}

\noindent \textbf{(C3)} There exists a Lipschitz constant $L_2 > 0$ such that:
\begin{align*}
\| f_i(x_i, \mathbf{x}_j) - f_i(y_i, \mathbf{x}_j) \| \leq L_2 \|(x_i, \mathbf{x}_j) - (y_i, \mathbf{x}_j) \|, \forall \ i \in \mathcal{I}, x_i, y_i \in \mathbb{R}^n, \mathbf{x}_j, \mathbf{y}_j \in \mathbb{R}^{N_i n}.
\end{align*}

\noindent From \eqref{eq: system} and \eqref{eq:boskos_dynamics} we get $f_i(x_i, \mathbf{x}_j) = \displaystyle -\sum_{j \in \mathcal{N}(i)}^{}(x_i-x_j)$. By checking all the conditions one by one for $f_i(x_i, \mathbf{x}_j)$ as in \eqref{eq: system}, it can be shown that our system satisfies all the conditions \textbf{(C1)}-\textbf{(C3)}. The proof can be found in Appendix \ref{app:suff_conditions}.

Based on the sufficient condition for well posed abstractions in \cite{boskos_cdc_2015}, the diameter $\bar{d}_{\max}$ and the time step $\delta t$ of the discretization $\bar{S}, \delta t$ can be selected as:
\begin{subequations}
	\begin{align}
	\bar{d}_{\max} & \in \left(0,\frac{(1-\lambda)^2 v_{\max}^{2}}{4ML}\right], \label{dmax} \\
	\delta t \in \Bigg[ & \frac{(1-\lambda)v_{\max}-\sqrt{(1-\lambda)^2 v_{\max}^{2}-4ML \bar{d}_{\max}}}{2ML}, \notag \\
	 & \frac{(1-\lambda)v_{\max}+\sqrt{(1-\lambda)^2 v_{\max}^{2}-4ML \bar{d}_{\max}}}{2ML} \Bigg], \label{deltat}
	\end{align}
\end{subequations}

where $L=\max\{3L_{2}+4L_{1}\sqrt{N_i},i\in\mathcal{I}\}$ and with the dynamics bound $M$ and the Lipschitz constants $L_1$, $L_2$ as previously defined. Furthermore, $\lambda\in(0,1)$ is a design parameter which quantifies the part of the free input that is additionally exploited for reachability purposes. In particular, given an agent's initial cell configuration, the agent can reach any point inside an appropriate ball at $\delta t$  through a parameterized feedback law in place of the free input $v_i$.  The radius of this ball increases proportionally to the value of $\lambda$, and thus, also the number of the agent's successor cells, which are the ones intersecting the ball. It is noted that an increasing choice of $\lambda$ results in finer discretizations, therefore providing a quantifiable trade-off between the discrete model's accuracy and complexity. Furthermore, it follows from the acceptable values of $\bar d_{\max}$ that the cells can be selected coarser, when (i) the available control $v_{\max}$ is larger, and, (ii) the coupling term bound $M$ together with the dynamics' variation, which is captured through the parameter $L$, are smaller. Analogous restrictions need to hold for the time step $\delta t$. In particular, the time step cannot be selected very large, because the required  control for the manipulation of the coupling terms increases due to the evolution of the agent's neighbors during the transition interval. Finally, the time step cannot be selected very small either, because controlling the agent to the same point from each initial condition in its cell, will require a large control effort over a very short transition interval.

\begin{remark} \label{remark:d_max_remark}
	Assume that a cell-decomposition of diameter $\bar{d}_{\text{max}}$ and a time step $\delta t$ which guarantee well-posed transitions, namely, which satisfy \eqref{dmax} and \eqref{deltat}, have been chosen. Then, it is also possible to choose any other cell-decomposition with diameter $\hat{d}_{\text{max}} \leq \bar{d}_{\text{max}}$ since, by \eqref{deltat}, the range of acceptable $\delta t$ increases.
\end{remark}

Having shown that the dynamics of system \eqref{eq: system} satisfy the sufficient conditions \textbf{(C1)}-\textbf{(C3)}, a well-posed space-time discretization $(\bar{S}, \delta t)$ has been obtained. Recall now Assumption \ref{assumption: AP_cell_decomposition}. It remains to establish the compliance of the cell decomposition $S = \{S_\ell\}_{\ell \in \mathbb{I}}$, which is given in the statement of Problem \ref{problem: basic_prob}, with the cell decomposition $\bar{S} = \{\bar{S}_l\}_{l \in \bar{\mathbb{I}}}$, which is the outcome of the abstraction. By the term of compliance, we mean that:
\begin{equation*} \label{eq:cell_decomposition_compliance}
\bar{S}_l \cap S_\ell \in S \cup \{\emptyset\}, \forall \ \bar{S}_l \in \bar{S}, S_\ell \in S, l \in \bar{\mathbb{I}}, \ell \in \mathbb{I}.
\end{equation*} 
In order to address this problem, define:
\begin{equation*} \label{eq:final_decomposition}
\hat{S} = \{ \hat{S}_{\hat{l}} \}_{\hat{l} \in \hat{\mathbb{I}}} = \{ \bar{S}_l \cap S_\ell : \ l \in \bar{\mathbb{I}}, \ell \in \mathbb{I} \} \backslash \{ \emptyset \},
\end{equation*}
which forms a cell decomposition and is compliant with the cell decomposition $S$ from Problem \ref{problem: basic_prob} with diameter $\hat{d}_{\text{max}} = \text{sup} \{ \|x - y\| : x,y \in \hat{S}_{\hat{l}}, \hat{l} \in \hat{\mathbb{I}} \} \leq \bar{d}_{\text{max}}$ and serves as the abstraction solution of this problem. %This can be deducted as follows: First, the index $\hat{\mathbb{I}}$ is used for the enumeration of the cells of $\hat{S}$. 
By taking all the intersections $\bar{S}_l \cap S_\ell, \forall \ l \in \bar{\mathbb{I}}, \forall \ \ell \in \mathbb{I}$ and enumerating them through the index set $\hat{\mathbb{I}}$, the cells $\{ \hat{S}_{\hat{l}} \}_{\hat{l} \in \hat{\mathbb{I}}}$ are constructed. For the cells $\{ \hat{S}_{\hat{l}} \}_{\hat{l} \in \hat{\mathbb{I}}}$, the following holds: $\forall \ \ell \in \mathbb{I}, \exists \ l \in \bar{\mathbb{I}}$ such that $\hat{S}_{\hat{l}} = \bar{S}_l \cap S_\ell$ and $\text{int}(\hat{S}_{\hat{l}}) \cap \text{int}(\hat{S}_{\hat{l}'}) \neq \emptyset$ for all $\hat{l}' \in \hat{\mathbb{I}} \backslash \{\hat{l}\}$. After all the intersections we have $\cup_{\hat{l} \in \hat{\mathbb{I}}} \hat{S}_{\hat{l}} = \mathcal{X}$. The diameter of the cell decomposition $\hat{S} =  \{ \hat{S}_{\hat{l}} \}_{\hat{l} \in \hat{\mathbb{I}}}$ is defined as $\hat{d}_{\text{max}} = \text{sup} \{ \|x - y\| : x,y \in \hat{S}_{\hat{l}}, \hat{l} \in \hat{\mathbb{I}} \} \leq \bar{d}_{\text{max}}$. Hence, according to Remark \ref{remark:d_max_remark}, we have a well-posed abstraction. The following Example is an illustration of these derivations. 

\begin{example} \label{ex: example_03}
	Let $S = \{S_\ell\}_{\ell \in \{1,\ldots,6\}}$ be the cell decomposition of Problem \ref{problem: basic_prob}, which is depicted in Figure \ref{fig: example_03} by the red rectangles. In the same figure, we illustrate the cell decomposition $\bar{S} = \{\bar{S}_l\}_{l \in \bar{\mathbb{I}} = \{1, \ldots, 6\}}$. $\bar{S}$ serves as potential solution of this Problem satisfying all the abstraction properties that have been mentioned in this Section. It can be observed that the two cell decompositions are not compliant according to \eqref{eq:cell_decomposition_compliance}. However, by using the methodology below Remark \ref{remark:d_max_remark}, a new cell decomposition $\hat{S} = \{ \hat{S}_{\hat{l}} \}_{\hat{l} \in \hat{\mathbb{I}} = \{1, \ldots, 15\}}$, which is compliant with $S$ and has $15$ regions, can be obtained. $\hat{S}$ forms the final cell decomposition solution of Problem \ref{problem: basic_prob} and is depicted in Figure \ref{fig:example_04}. Let also $\bar{d}_{\text{max}}, \hat{d}_\text{max}$ be the diameters of the cell decompositions $\bar{S}, \hat{S}$ respectively. Then, it holds that $\hat{d}_\text{max} \leq \bar{d}_{\text{max}}$, which is in accordance with Remark \ref{remark:d_max_remark}.
	
	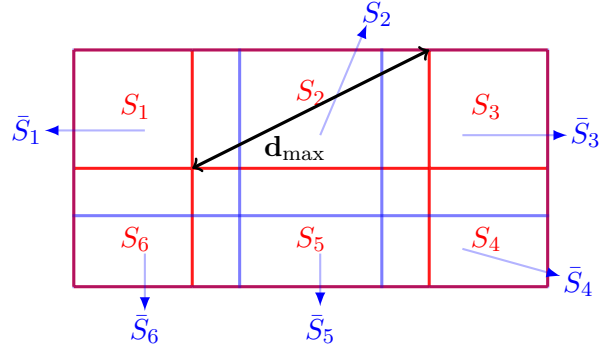
\begin{figure}[t!]
		\centering
		\begin{tikzpicture}[scale = 0.63]
		
		% plot the red grid		
		\draw[step=2.5, line width=.04cm, color = red, draw opacity = 0.9] (-2.5, -5.0) grid (0,0);
		\draw[line width=.04cm, color = red, draw opacity = 0.9] (-7.5,0.0) -- (-2.5,0.0);
		\draw[line width=.04cm, color = red, draw opacity = 0.9] (-7.5,-2.5) -- (-2.5,-2.5);
		\draw[line width=.04cm, color = red, draw opacity = 0.9] (-7.5,-5.0) -- (-2.5,-5.0);
		\draw[step=2.5, line width=.04cm, color = red, draw opacity = 0.9] (-10.0, -5.0) grid (-7.5,0);
		
		% plot the blue grid
		\draw[line width=.04cm, color = blue, draw opacity = 0.3] (0,0) rectangle (-3.5,-3.5);
		\draw[line width=.04cm, color = blue, draw opacity = 0.3] (0,-3.5) rectangle (-3.5,-5.0);
		\draw[line width=.04cm, color = blue, draw opacity = 0.3] (-3.5,-3.5) rectangle (-6.5,-5);
		\draw[line width=.04cm, color = blue, draw opacity = 0.3] (-3.5,0) rectangle (-6.5,-3.5);
		\draw[line width=.04cm, color = blue, draw opacity = 0.3] (-6.5,0) rectangle (-10.0,-3.5);
		\draw[line width=.04cm, color = blue, draw opacity = 0.3] (-6.5,-3.5) rectangle (-10.0,-5.0);
		
		% red cell decomposition
		\node[color = red] at (-8.7, -1.2) {$S_1$};
		\node[color = red] at (-5.0, -0.9) {$S_2$};
		\node[color = red] at (-1.3, -1.2) {$S_3$};
		\node[color = red] at (-1.3, -4.0) {$S_4$};
		\node[color = red] at (-5.0, -4.0) {$S_5$};
		\node[color = red] at (-8.7, -4.0) {$S_6$};
		
		% blue cell decomposition
		\draw [-latex, color = blue, draw opacity = 0.3, thick, shorten >= 1.00cm] (-1.8, -1.8) -- (2.0, -1.8);
		\node[color = blue, draw opacity = 0.3] at (0.80, -1.75) {$\bar{S}_3$};
		
		\draw [-latex, color = blue, draw opacity = 0.3, thick, shorten >= 1.00cm] (-1.8, -4.2) -- (1.8, -5.2);
		\node[color = blue, draw opacity = 0.3] at (0.65, -4.9) {$\bar{S}_4$};
		
		\draw [-latex, color = blue, draw opacity = 0.3, thick, shorten >= 1.00cm] (-4.8, -1.8) -- (-3.2, 2.0);
		\node[color = blue, draw opacity = 0.3] at (-3.6, 0.8) {$\bar{S}_2$};
		
		\draw [-latex, color = blue, draw opacity = 0.3, thick, shorten >= 1.00cm] (-4.8, -4.3) -- (-4.8, -7.0);
		\node[color = blue, draw opacity = 0.3] at (-4.8, -5.9) {$\bar{S}_5$};
		
		\draw [-latex, color = blue, draw opacity = 0.3, thick, shorten >= 1.00cm] (-8.5, -4.3) -- (-8.5, -7.1);
		\node[color = blue, draw opacity = 0.3] at (-8.5, -5.9) {$\bar{S}_6$};
		
		\draw [-latex, color = blue, draw opacity = 0.3, thick, shorten >= 1.00cm] (-8.5, -1.7) -- (-12.2, -1.7);
		\node[color = blue, draw opacity = 0.3] at (-11, -1.7) {$\bar{S}_1$};
		
		% plot d_max
		\draw[color=black,very thick,<->] (-2.5,0) -- (-7.5,-2.5);
		\coordinate [label=below right:$\bf{d_{\max}}$] (A) at (-6.2,-1.6);
		
		\end{tikzpicture}
		\caption{An example with a given cell decomposition $S = \{S_l\}_{l \in \{1,\ldots,6\}}$ from Problem \ref{problem: basic_prob} and a non-compliant cell decomposition $\bar{S} = \{\bar{S}_l\}_{l \in \bar{\mathbb{I}} = \{1, \ldots, 6\}}$ which is the outcome of the proposed abstraction technique.}
		\label{fig: example_03}
	\end{figure}
	
	\begin{figure}[t!]
		\centering
		\begin{tikzpicture}[scale = 0.63]
		
		% plot the blakc grid		
		\draw[step=2.5, line width=.04cm, color = black, draw opacity = 0.9] (-2.5, -5.0) grid (0,0);
		\draw[line width=.04cm, color = black, draw opacity = 0.9] (0,0) rectangle (-3.5,-3.5);
		\draw[line width=.04cm, color = black, draw opacity = 0.9] (-7.5,0.0) -- (-2.5,0.0);
		\draw[line width=.04cm, color = black, draw opacity = 0.9] (-7.5,-2.5) -- (-2.5,-2.5);
		\draw[line width=.04cm, color = black, draw opacity = 0.9] (-7.5,-5.0) -- (-2.5,-5.0);
		\draw[step=2.5, line width=.04cm, color = black, draw opacity = 0.9] (-10.0, -5.0) grid (-7.5,0);
		
		\draw[line width=.04cm, color = black, draw opacity = 0.9] (0,-3.5) rectangle (-3.5,-5.0);
		\draw[line width=.04cm, color = black, draw opacity = 0.9] (-3.5,-3.5) rectangle (-6.5,-5);
		\draw[line width=.04cm, color = black, draw opacity = 0.9] (-3.5,0) rectangle (-6.5,-3.5);
		\draw[line width=.04cm, color = black, draw opacity = 0.9] (-6.5,0) rectangle (-10.0,-3.5);
		\draw[line width=.04cm, color = black, draw opacity = 0.9] (-6.5,-3.5) rectangle (-10.0,-5.0);
		
		%  fill the rectangulars
		\draw[fill = green!5] (0,0) rectangle (-2.5,-2.5);
		\draw[fill = green!5] (0,-2.5) rectangle (-2.5,-3.5);
		\draw[fill = green!5] (0,-3.5) rectangle (-2.5,-3.5);
		\draw[fill = green!5] (0,-3.5) rectangle (-2.5,-5.0);
		\draw[fill = green!5] (-2.5,0) rectangle (-3.5,-2.5);
		\draw[fill = green!5] (-2.5,-2.5) rectangle (-3.5,-3.5);
		\draw[fill = green!5] (-2.5,0) rectangle (-3.5,-2.5);
		\draw[fill = green!5] (-2.5,-3.5) rectangle (-3.5,-5.0);
		\draw[fill = green!5] (-3.5,0.0) rectangle (-6.5,-2.5);
		\draw[fill = green!5] (-6.5,0.0) rectangle (-7.5,-2.5);
		\draw[fill = green!5] (-7.5,0.0) rectangle (-7.5,-2.5);
		\draw[fill = green!5] (-7.5,0.0) rectangle (-10.0,-2.5);
		\draw[fill = green!5] (-3.5,-2.5) rectangle (-6.5,-3.5);
		\draw[fill = green!5] (-6.5,-2.5) rectangle (-7.5,-3.5);
		\draw[fill = green!5] (-7.5,-2.5) rectangle (-10.0,-3.5);
		\draw[fill = green!5] (-3.5,-3.5) rectangle (-6.5,-5.0);
		\draw[fill = green!5] (-6.5,-3.5) rectangle (-7.5,-5.0);
		\draw[fill = green!5] (-7.5,-3.5) rectangle (-10.0,-5.0);
		
		% red cell decomposition
		\node[color = black] at (-8.9, -0.8) {$\hat{S}_1$};
		\node[color = black] at (-7.0, -1.2) {$\hat{S}_2$};
		\node[color = black] at (-5.0, -1.2) {$\hat{S}_3$};
		\node[color = black] at (-3.0, -1.2) {$\hat{S}_4$};
		\node[color = black] at (-1.3, -1.2) {$\hat{S}_5$};
		\node[color = black] at (-1.3, -2.95) {$\hat{S}_6$};
		\node[color = black] at (-3.0, -2.95) {$\hat{S}_7$};
		\node[color = black] at (-5.0, -2.95) {$\hat{S}_8$};
		\node[color = black] at (-7.0, -2.95) {$\hat{S}_{9}$};
		\node[color = black] at (-8.7, -2.95) {$\hat{S}_{10}$};
		\node[color = black] at (-8.7, -4.20) {$\hat{S}_{11}$};
		\node[color = black] at (-7.0, -4.20) {$\hat{S}_{12}$};
		\node[color = black] at (-5.0, -4.20) {$\hat{S}_{13}$};
		\node[color = black] at (-3.0, -4.20) {$\hat{S}_{14}$};
		\node[color = black] at (-1.3, -4.20) {$\hat{S}_{15}$};
		
		% plot d_max
		\draw[color=red,very thick,<->] (-7.5,0) -- (-10.0,-2.5);
		\node[color = red] at (-8.8, -2.10) {$\hat{d}_{\text{max}}$};
		
		\end{tikzpicture}
		\caption{The resulting compliant cell decomposition $\hat{S} =  \{ \hat{S}_{\hat{l}} \}_{\hat{l} \in \hat{\mathbb{I}} = \{1, \ldots, 15\}}$ of the Example \ref{ex: example_03} which serves as solution to Problem 1.}
		\label{fig:example_04}
	\end{figure}
\end{example}

\subsubsection{Discrete System Abstraction} \label{sec:discrete_system_abstraction}

For the solution to Problem \ref{problem: basic_prob}, the WTS of this agent which corresponds to the cell decomposition $\hat{S}$ with diameter $\hat{d}_{\text{max}}$ and the time step $\delta t$ will be exploited. Thus, the WTS of each agent is defined as follows:
\begin{definition} \label{def: indiv_WTS}
	The motion of each agent $i \in \mathcal{I}$ in the workspace is modeled by the WTS $\ \mathcal{T}_i = (S_i, S_i^{\text{init}}, Act_i, \longrightarrow_i, d_i, AP_i, \hat L_i)$ where:
	\begin{itemize}
		\item $S_i = \hat{\mathbb{I}}$ is the set of states of each agent which is the set of indices of the cell decomposition.
		\item $S_i^{\text{init}} \subseteq S_i$ is a set of initial states defined by the agents' initial positions in the workspace.
		\item $Act_i = {\hat{\mathbb{I}}}^{N_i+1}$, the set of actions representing where agent $i$ and its neighbors are located.
		\item For a pair $(l_i, {\bf{l}_i}, l'_i)$ we have that $(l_i, {\bf{l}_i}, l'_i) \in \longrightarrow_i$ iff $l_i \overset{\bf{l}_i}{\longrightarrow_i} l_i'$ is well-posed for each $l_i, l'_i \in S_i$ and ${\bf{l}_i} = (l_i, l_{j_1}, \ldots, l_{j_{N_i}}) \in Act_i$.
		\item $d_i: \longrightarrow_i \rightarrow \mathbb{T}$, is a map that assigns a positive weight (duration) to each transition. The duration of each transition is exactly equal to $\delta t > 0$.
		\item $\AP_i = \Sigma_i$, is the set of atomic propositions which are inherent properties of the workspace.
		\item $L_i: S_i \rightarrow 2^{AP_i}$, is the labeling function that maps every state $s \in S_i$ into the services that can be provided in this state.
	\end{itemize}
\end{definition}

The individual WTSs of the agents will allow us to work completely in the discrete level and design sequences of controllers that solve Problem \ref{problem: basic_prob}.

Every WTS $\mathcal{T}_i, i \in \mathcal{I}$ generates timed runs and timed words of the form $r_i^t = (r_i(0)$, $\tau_i(0))(r_i(1)$, $\tau_i(1))(r_i(2)$, $\tau_i(2)) \ldots$, $w_i^t = (L_i(r_i(0))$, $\tau_i(0))(L_i(r_i(1))$, $\tau_i(1))(L_i(r_i(2))$, $\tau_i(2)) \ldots$ respectively, over the set $2^{AP_i}$ according to Def. \ref{run_of_WTS} with $\tau_i(j) = j \delta t, \forall \ j \ge 0$. The relation between the timed words that are generated by the WTSs $\mathcal{T}_i, i \in \mathcal{I}$ with the timed service words produced by the trajectories $x_i(t), i \in \mathcal{I}, t \ge 0$ is provided through the following remark:

\begin{remark} \label{lemma:compliant_WTS_runs_with_trajectories}
	By construction, each timed word produced by the WTS $\mathcal{T}_i$ is a service timed word associated with the trajectory $x_i(t)$ of the system \eqref{eq: system}. Hence, if we find a timed word of $\mathcal{T}_i$ satisfying a formula $\varphi_i$ given in MITL, we also find for each agent $i$ a desired timed word of the original system, and hence trajectories $x_i(t)$ that are a solution to the Problem \ref{problem: basic_prob}. Therefore, the produced timed words of $\mathcal{T}_i$ are compliant with the service timed words of the trajectories $x_i(t)$.
\end{remark}

\subsection{Runs Consistency} \label{sec: runs consistency}

Due to the coupled dynamics between the agents, it is required that each individual agent's run is compliant with the corresponding discrete trajectories of its neighbors, which determine the actions in the agent's run. Therefore, even though we have the individual WTS of each agent, the runs that the latter generates may not be performed by an agent due to the constrained motion that is imposed by the coupling terms. Hence, we need to synchronize the agents at each time step $\delta t$ and determine which of the generated runs of the individual WTS can be performed by the agent. Hereafter, they will be called \textit{consistent runs}. In order to address the aforementioned issue, we provide a centralized product WTS which captures the behavior of the coupled multi-agent system as a team, and the generated product run (see Def. \ref{def: consistent_runs}) can later be projected onto consistent individual runs. The following two definitions deal with the product WTS and consistent runs respectively.

\begin{definition} \label{def:product_TS}
	Given the individual WTSs $\mathcal{T}_i, i \in \mathcal{I}$ from Def. \ref{def: indiv_WTS}, the product WTS $\mathcal{T}_p = (S_p, S_p^{\text{init}}, \longrightarrow_p, L_p)$ is defined as follows:
	\begin{itemize}
		\item $S_p = \hat{\mathbb{I}}^N$;
		\item $(s_1, \ldots, s_N) \in S^{\text{init}}$ if $s_i \in S_i^{\text{init}}, \forall \ i \in \mathcal{I}$;
		\item $(\bf{l}, \bf{l}') \in \longrightarrow_p$ iff $l_i' \in \text{Post}_i(l_i, \text{pr}_i(\bf{l})), \forall \ i \in \mathcal{I}, \forall \ \bf{l} = (l_1, \ldots, l_N), \bf{l}' = (l'_1, \ldots, l'_N)$;
		\item $L_p:\hat{\mathbb I}^N\to 2^{\cup_i^N\Sigma_i}$ defined as $L_p({\bf l})=\cup_{i=1}^NL_i(l_i)$;
		\item $d_p: \longrightarrow_p \rightarrow \mathbb{T}$: as in the individual WTS's case, with transition weight $d_p(\cdot) = \delta t$.
	\end{itemize}
\end{definition}

\begin{definition} \label{def: consistent_runs}
	Given a timed run:
	\begin{align*}
	&r_p^t = ((r_p^1(0), \ldots , r_p^N(0)),\tau_p(0)) \noindent ((r_p^1(1),\ldots,r_p^N(1)),\tau_p(1))\ldots, \notag
	\end{align*}
	that is generated by the product WTS $\mathcal T_p$, the induced set of projected runs $$\{r_i^t = (r_p^i(0), \tau_p(0))(r_p^i(1), \tau_p(1)) \ldots  : i \in \mathcal I\},$$ of the WTSs $\mathcal T_1,\ldots, \mathcal T_N$, respectively will be called \textit{consistent runs}. Since the duration of each agent's transition is $\delta t$ it holds that $\tau_p(j) = j \delta t, j \geq 0$.
\end{definition}

Therefore, through the product WTS $\mathcal{T}_p$, we can always generate individual consistent runs for each agent. It remains to provide a systematic approach of how to determine consistent runs $\widetilde{r}_1, \ldots, \widetilde{r}_N$  which are associated with the corresponding time serviced words $\widetilde{w}_1^t, \ldots, \widetilde{w}_N^t$. Note that we use the tilde accent to denote timed runs and words that correspond to the problem solution. The corresponding compliant trajectories $x_1(t), \ldots, x_N(t)$ of the timed words $\widetilde{w}_1^t, \ldots, \widetilde{w}_N^t$ satisfy the corresponding MITL formulas $\varphi_1, \ldots, \varphi_N$, and they are a solution to Problem \ref{problem: basic_prob}. This follows from the fact that the product transition system is simulated by the $\delta t$-sampled version of the continuous system (see \cite{tabuada_book_verification} for the definition of a simulation relation). In particular, let $\mathcal{T}_{\delta t}$ be the $\delta t$-sampled WTS of system (1), as defined in \cite[Def. 11.4]{tabuada_book_verification}, with labeling function $L_{\delta t}:\mathbb R^{Nn}\to 2^{\cup_{i=1}^N\Sigma_i}$ given as $L_{\delta t}(x_1,\ldots,x_N)=\cup_{i=1}^N \Lambda_i(x_i)$ and $\Lambda_i$ as defined in Section 2. Consider also the WTS $\mathcal{T}_{p}$ and the relation $\mathcal{R} \subseteq S_{p} \times \mathcal{X}^N$ given as $({\bf l},(x_1,\ldots,x_N))\in\mathcal{R}$, iff $(x_1,\ldots,x_N)\in S_{l_1}\times\cdots\times S_{l_N}$, where ${\bf l}=(l_1,\ldots,l_N)$. Then,  from the definition of the agent's individual transitions in each WTS $\mathcal T_i$ and the fact that for all points in a cell the same atomic proposition hold true, it can be deduced that $\mathcal{R}$ is a simulation relation from $\mathcal{T}_{p}$ to the $\delta t$-sampled WTS $\mathcal{T}_{\delta t}$.

\begin{remark}
	We chose to utilize decentralized abstractions, to generate the individual WTSs $\mathcal{I}, i \in \mathcal{I}$ for each agent and to compute the synchronized-centralized product WTS $\mathcal{T}_p$ for the following reasons: 
	\begin{enumerate}
		\item The state space of the centralized system to be abstracted is $\mathcal{X}^N \subseteq \mathbb{R}^{Nn}$, which is harder to visualize and handle as well as not naturally related to the individual specifications. Thus, it is more “natural” to define the specifications through the individual transition system of each agent corresponding to a discretization of $\mathcal{X}$ and then generate the product in order to obtain potential consistent satisfying plans.
		\item Additionally, many centralized abstraction frameworks are based on approximations of the system's reachable sets from a given cell
		over the transition time interval. These frameworks, require in the general nonlinear case, global dynamics properties and may avoid taking into account the finer dynamics properties of the individual entities, which can lead to more conservative estimates for large scale systems. 
	\end{enumerate}
\end{remark}

We provide here an example that explains the notation that has been introduced until now.
\begin{example} \label{ex: example_2}
	\begin{figure}[t!]
		\centering
		\begin{tikzpicture}[scale = 0.76]
		
		% plot the AP
		\filldraw[fill=yellow!40, line width=.04cm] (0, 0) rectangle +(3, 3.0);
		\filldraw[fill=orange!40, line width=.04cm] (-7.5, -3) rectangle +(3, 3.0);
		\filldraw[fill=red!20, line width=.04cm] (-4.5, 0) rectangle +(4.5, 3.0);
		\filldraw[fill=black!10, line width=.04cm] (-7.5, 0) rectangle +(3, 3.0);
		\filldraw[fill=blue!20, line width=.04cm] (0, -3) rectangle +(3, 3.0);
		\filldraw[fill=green!40, line width=.04cm] (-4.5, -3) rectangle +(4.5, 3);
		
		% plot the grid
		%\draw[step=1.5, line width=.04cm] (-7.5,-1.5) grid (4,3);
		\draw[step=1.5, line width=.04cm] (-7.5, -3) grid (3,3);
		
		% plot the robot
		%\filldraw[fill=red!70, line width=.04cm]  (-6.8,-0.80) circle (0.25cm);
		%\node at (-7,-1.20) {$i$};

		% plot nodes
		\draw (-6.8,-0.6) node[circle, inner sep=0.8pt, fill=black, label={below:{$i$}}] (A) {};
		\draw (-3.8, 0.7) node[circle, inner sep=0.8pt, fill=black, label={below:{$ $}}] (B) {};
		\draw (-0.8, -0.7) node[circle, inner sep=0.8pt, fill=black, label={below:{$ $}}] (C) {};
		\draw (0.8, 0.7) node[circle, inner sep=0.8pt, fill=black, label={below:{$ $}}] (D) {};
		\draw (-5.3, -2.1) node[circle, inner sep=0.8pt, fill=black, label={below:{$j_2$}}] (E) {};
		\draw (-6.8, 2.3) node[circle, inner sep=0.8pt, fill=black, label={below:{$ $}}] (F) {};
		\draw (-5.0, 2.2) node[circle, inner sep=0.8pt, fill=black, label={below:{$ $}}] (G) {};
		\draw (-0.8, 2.2) node[circle, inner sep=0.8pt, fill=black, label={below:{$ $}}] (H) {};
		\draw (2.2, 2.2) node[circle, inner sep=0.8pt, fill=black, label={below:{$ $}}] (I) {};
		\draw (-3.8, -0.85) node[circle, inner sep=0.8pt, fill=black, label={below:{$ $}}] (K) {};
		\draw (-0.8, -2.30) node[circle, inner sep=0.8pt, fill=black, label={below:{$ $}}] (L) {};
		\draw (0.75, -0.8) node[circle, inner sep=0.8pt, fill=black, label={below:{$ $}}] (M) {};
		
		% plot the neighbors
		
		\node at (-6.8, 2.69) {$j_1$};

		% plot transitions
		
		\draw[->, red, line width=.04cm] (A) to [bend left=35] (B);
		\draw[->, red, line width=.04cm] (B) to [bend left=-35] (C);
		\draw[->, red, line width=.04cm] (C) to [bend left=35] (D);
		\draw[->, red, line width=.04cm] (F) to [bend left=15] (G);
		\draw[->, red, line width=.04cm] (G) to [bend left=15] (H);
		\draw[->, red, line width=.04cm] (H) to [bend right=15] (I);
		\draw[->, red, line width=.04cm] (E) to [bend right=15] (K);
		\draw[->, red, line width=.04cm] (K) to [bend right=15] (L);
		\draw[->, red, line width=.04cm] (L) to [bend right=15] (M);
		
		% plot the neighbor
		
		\draw[color = black,very thick,<->] (1.5,-3) -- (3,-1.5);
		\coordinate [label=below right:$\bf{d_{\max}}$] (N) at (2.2,-2);
		
		% regions notation
		\node at (1.8, 2.7) {$22$};
		\node at (1.9, -0.3) {$8$};
		\node at (-7.2, 1.20) {$15$};
		\node at (-7.2, -1.9) {$1$};
		%\node at (-7.2, 2.7) {$28$};
		
		% cell notation
		\node [red!40] at (-2.3, 3.5) {$S_5$};
		\node [yellow!100] at (1.7, 3.4) {$S_4$};
		\node [blue!20] at (1.7, -3.7) {$S_3$};
		\node [green!70] at (-2.3, -3.7) {$S_2$};
		\node [orange!60] at (-6.3, -3.7) {$S_1$};
		\node [black!30] at (-6, 3.5) {$S_6$};
		
		% dt
		
		\node [red] at (-5.6, 0.95) {$\delta t$};
		%\node [red] at (-2.5, -1.15) {$\delta t$};
		%\node [red] at (-0.5, 0.5) {$\delta t$};
		
		% draw dashed line
		\draw [dashed, black] (A) -- (E);
		\draw [dashed, black] (A) -- (F);
		\draw [dashed, black] (B) -- (G);
		\draw [dashed, black] (C) -- (H);
		\draw [dashed, black] (D) -- (I);
		\draw [dashed, black] (K) -- (B);
		\draw [dashed, black] (L) -- (C);
		\draw [dashed, black] (M) -- (D);
		\end{tikzpicture}
		\caption{Timed runs of the agents $i, j_1, j_2$}
		\label{fig: example_1}
	\end{figure}
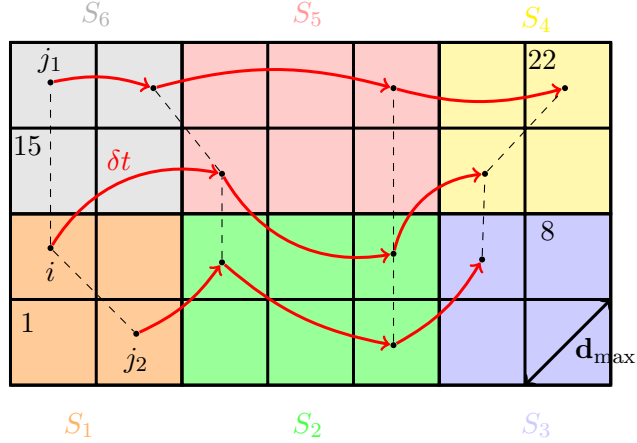
	
	Consider $N = 3$ agents performing in the workspace with $\mathcal{N}(i) = \{1,2\}$ as depicted in Figure \ref{fig: example_1}. $S = \{S_\ell\}_{\ell \in \mathbb{I} = \{1, \ldots, 6\}}$ is the given cell decomposition from Problem \ref{problem: basic_prob} and ${S} = \{\bar{S}_{l}\}_{l \in \mathbb{I} = \{1, \ldots, 28\}}$ is the cell decomposition which is the outcome of the proposed abstraction technique. Let the atomic propositions be $\{p_1,\ldots,p_6\} = \{\rm orange, $ $green, blue, yellow, red, grey\}$. The red arrows represent both the transitions of the agent $i$ and its neighbors. The dashed lines indicate the edges in the network graph. For the atomic propositions we have that $L_i(14) = \{p_1\}, L_i(17) = \{p_5\}, L_i(10) = \{p_2\}, L_i(20) = \{p_4\}, L_{j_1}(28) = \{p_6\} = L_{j_1}(27), L_{j_1}(24) = \{p_5\}, L_{j_1}(22) = \{p_4\}, L_{j_2}(2) = \{p_1\}, L_{j_2}(12) = \{p_2\} = L_{j_2}(5), L_{j_2}(9) = \{p_3\}$. Note also the diameter of the cells $\hat{d}_{\text{max}} = \bar{d}_{\text{max}}$. For the cell configurations we have:
	\begin{align*}
	Init \ (t = 0)&:
	\begin{cases}
	{\bf{l}_i} = (14, 28, 2), \\
	{\bf{l}_{j_1}} = (28, 14), \\
	{\bf{l}_{j_2}} = (2, 14), \\
	\end{cases} 
	Step \ 1 \ (t = \delta t) :
	\begin{cases}
	{\bf{l}_i} = (17, 27, 13), \\
	{\bf{l}_{j_1}} = (27, 17), \\
	{\bf{l}_{j_2}} = (13, 17), \\
	\end{cases} \notag \\
	Step \ 2 \ (t = 2 \delta t)&:
	\begin{cases}
	{\bf{l}_i} = (10, 24, 5), \\
	{\bf{l}_{j_1}} = (24, 10), \\
	{\bf{l}_{j_2}} = (5, 10), \\
	\end{cases} 
	Step \ 3 \ (t = 3 \delta t) :
	\begin{cases}
	{\bf{l}_i} = (20, 22, 9), \\
	{\bf{l}_{j_1}} = (22, 20), \\
	{\bf{l}_{j_2}} = (9, 20), \\
	\end{cases} \notag
	\end{align*}
	which are actions to the corresponding transitions. Three consistent timed runs are given as:
	\begin{align}
	r_i^t & = (r_i(0) = 14, \tau_i(0) = 0) (r_i(1) = 17, \tau_i(1) = \delta t)(r_i(2) = 10, \tau_i(2) = 2 \delta t) \notag \\ 
	& \hspace{60mm} (r_i(3) = 20, \tau_i(3) = 3 \delta t), \notag \\
	r_{j_1}^t & = (r_{j_1}(0) = 28, \tau_{j_1}(0) = 0) (r_{j_1}(1) = 27, \tau_{j_1}(1) = \delta t) (r_{j_1}(2) = 24, \tau_{j_1}(2) = 2 \delta t) \notag \\ 
	& \hspace{60mm} (r_{j_1}(3) = 22, \tau_{j_1}(3) = 3\delta t), \notag \\
	r_{j_2}^t & = (r_{j_2}(0) = 2, \tau_{j_2}(0) = 0) (r_{j_2}(1) = 13, \tau_{j_2}(1) = \delta t)(r_{j_2}(2) = 5, \tau_{j_2}(2) = 2 \delta t) \notag \\ 
	& \hspace{60mm} (r_{j_2}(3) = 9, \tau_{j_2}(3) = 3\delta t). \notag
	\end{align}
	It  can be observed that $r_i^t \models (\varphi_i = \Diamond_{[0, 6]}\{yellow\})$ if $3 \delta t \in [0, 6]$, $r_{j_1}^t \models (\varphi_{j1} = \Diamond_{[3, 10]}\{red\})$ if $2 \delta t \in [3, 10]$ and $r_{j_2}^t \models (\varphi_{j2} = \Diamond_{[3, 9]}\{blue\})$ if $3 \delta t \in [3, 9]$. %For $\delta t = 1$, all the agents satisfy their goals.
\end{example}

\subsection{Controller Synthesis} \label{sec: synthesis}

The proposed controller synthesis procedure is described with the following steps:

\begin{enumerate}
	\item $N$ TBAs $\mathcal{A}_i, \ i \in \mathcal{I}$ that accept all the timed runs satisfying the corresponding specification formulas $\varphi_i, i \in \mathcal{I}$ are constructed.
	\item A B\"uchi WTS $\tilde{\mathcal T}_i = \mathcal{T}_i \otimes \mathcal{A}_i$ (see Def. \ref{def: buchi_WTS} below) for every $i \in \mathcal{I}$ is constructed. The accepting runs of $\tilde{\mathcal T}_i$, computed using standard graph search algorithms, are the individual runs of the $\mathcal{T}_i$ that satisfy the corresponding MITL formula $\varphi_i, \ i \in \mathcal{I}$.
	\item We pick a set of accepting runs $\{\widetilde{r}^t_1, \ldots, \widetilde{r}^t_N\}$ from Step~2. We check if they are consistent according to Def. \ref{def: consistent_runs}. If this is true then we proceed with Step 5. If this is not true then we repeat Step 3 with a different set of accepting runs. At worst case, we perform a finite predefined number of selections $R_{\text{selec}}$; if a consistent set of accepting runs is not found, we proceed with the less efficient centralized procedure in Step 4, which however searches through all sets of all possible accepting runs.
	\item We create the product $\widetilde{\mathcal{T}}_p = \mathcal{T}_p \otimes \mathcal{A}_p$ where $\mathcal{A}_p$ is the TBA that accepts all the words that satisfy the formula $\varphi = \varphi_1 \wedge \ldots \wedge \varphi_N$. An accepting run $\widetilde{r}_p$ of the product is projected into the accepting runs $\{\widetilde{r}_1, \ldots, \widetilde{r}_N\}$. If there is no accepting run found in  $\mathcal{T}_p \otimes \mathcal{A}_p$, then Problem \ref{problem: basic_prob} has no solution.	
	\item The abstraction procedure allows to find an explicit feedback law for each
	transition in $\mathcal T_i$. Therefore,
	an accepting run $\widetilde{r}^t_i$ in $\mathcal T_i$ that takes the form of a sequence of transitions is realized in the system in \eqref{eq: system} via the corresponding sequence of feedback laws.
	%(\textcolor{red}{From the sequence of states that are appearing in the run along with the corresponding Actions, there is a mapping from the actions of each transition to the control inputs $v_i$. Dimitris how can we write it better and compact?}).
\end{enumerate}

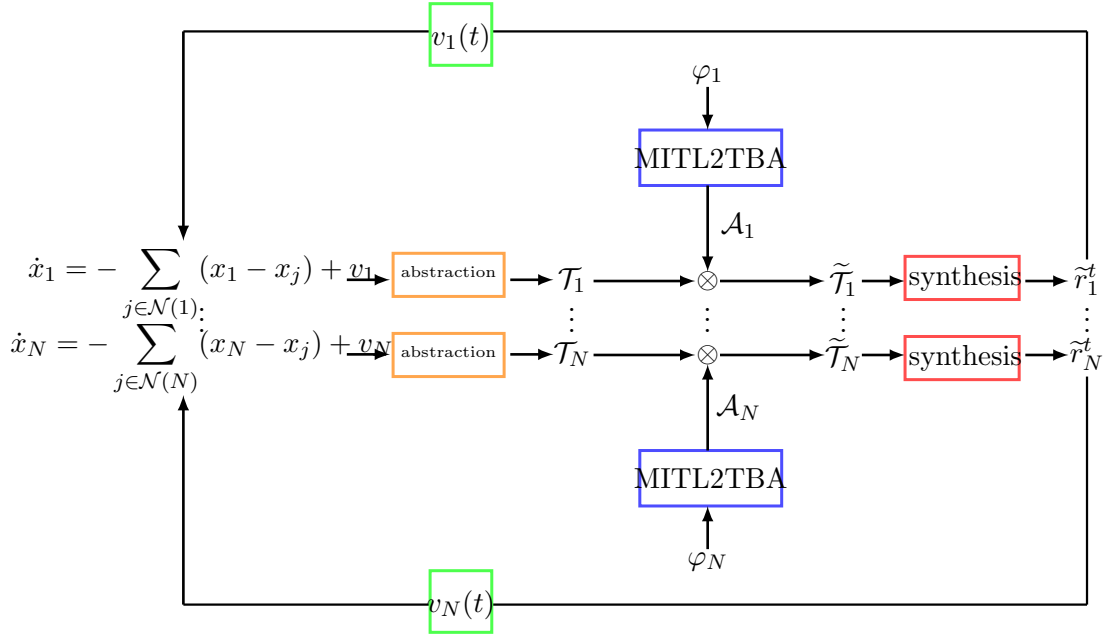
\begin{figure*}[t!]
	\centering
	\begin{tikzpicture}[scale = 0.82] 
	
	% convert formulas phi_i to buchis A_i
	\draw[blue!70, line width=.04cm] (-15.6, 6.0) rectangle +(2.3, 0.9);
	\node at (-14.45, 6.45) {$\text{MITL2TBA}$};
	\draw[blue!70, line width=.04cm] (-15.6, 0.80) rectangle +(2.3, 0.9);
	\node at (-14.45, 1.25) {$\text{MITL2TBA}$};
	
	\draw[-latex, draw=black, line width = 1.0] (-14.5,6.0) -- (-14.5,4.6);
	\draw[-latex, draw=black, line width = 1.0] (-14.5,7.6) -- (-14.5,6.9);
	\draw[-latex, draw=black, line width = 1.0] (-14.5,0.1) -- (-14.5,0.8);
	\draw[-latex, draw=black, line width = 1.0] (-14.5,1.7) -- (-14.5,3.1);
	
	\node at (-14.5, 7.80) {$\varphi_1$};
	\node at (-14.5, -0.1) {$\varphi_N$};
	
	\node at (-14.0, 5.35) {$\mathcal{A}_1$};
	\node at (-14.0, 2.40) {$\mathcal{A}_N$};
	
	\node at (-14.5, 3.25) {$\otimes$};
	\node at (-14.5, 4.45) {$\otimes$};
	
	\node at (-14.5, 3.95) {$\vdots$};
	
	% create \tilde T_i
	\draw[-latex, draw=black, line width = 1.0] (-14.30,4.45) -- (-12.6,4.45);
	\draw[-latex, draw=black, line width = 1.0] (-14.30,3.25) -- (-12.6,3.25);
	
	\node at (-12.3, 4.47) {$\widetilde{\mathcal{T}}_1$};
	\node at (-12.3, 3.95) {$\vdots$};
	\node at (-12.3, 3.26) {$\widetilde{\mathcal{T}}_N$};
	
	% create T_i
	\draw[-latex, draw=black, line width = 1.0] (-16.35,4.45) -- (-14.7,4.45);
	\draw[-latex, draw=black, line width = 1.0] (-16.35,3.25) -- (-14.7,3.25);
	
	\node at (-16.7, 4.47) {$\mathcal{T}_1$};
	\node at (-16.7, 3.95) {$\vdots$};
	\node at (-16.7, 3.26) {$\mathcal{T}_N$};
	
	% create box algorithm
	\draw[-latex, draw=black, line width = 1.0] (-12.00,4.45) -- (-11.3,4.45);
	\draw[-latex, draw=black, line width = 1.0] (-12.00,3.25) -- (-11.3,3.25);
	
	%\draw[red!70, line width=.04cm] (-11.2, 2.80) rectangle +(1.8, 2.2);
	%\node at (-10.30, 4.20) {$\text{graph}$};
	%\node at (-10.30, 3.85) {$\text{search}$};
	%\node at (-10.30, 3.40) {$\text{algorithm}$};
	
	\draw[red!70, line width=.04cm] (-11.3, 2.85) rectangle +(1.85, 0.70);
	\draw[red!70, line width=.04cm] (-11.3, 4.15) rectangle +(1.85, 0.70);
	\node at (-10.32, 4.50) {$\text{synthesis}$};
	\node at (-10.32, 3.20) {$\text{synthesis}$};
	
	% create runs
	\draw[-latex, draw=black, line width = 1.0] (-9.35,4.45) -- (-8.65,4.45);
	\draw[-latex, draw=black, line width = 1.0] (-9.35,3.25) -- (-8.65,3.25);
	
	\node at (-8.35, 4.47) {$\widetilde{r}_1^t$};
	\node at (-8.35, 3.95) {$\vdots$};
	\node at (-8.35, 3.26) {$\widetilde{r}_N^t$};
	
	% create abstractions box
	\draw[-latex, draw=black, line width = 1.0] (-17.70,4.47) -- (-17.00,4.47);
	\draw[-latex, draw=black, line width = 1.0] (-17.70,3.26) -- (-17.00,3.26);
	
	\draw[orange!70, line width=.04cm] (-19.6, 2.90) rectangle +(1.8, 0.7);
	\draw[orange!70, line width=.04cm] (-19.6, 4.20) rectangle +(1.8, 0.7);
	\node at (-18.70, 4.60) {$\tiny \text{abstraction}$};
	\node at (-18.70, 3.30) {$\tiny \text{abstraction}$};
	
	\draw[-latex, draw=black, line width = 1.0] (-20.35,4.47) -- (-19.65,4.47);
	\draw[-latex, draw=black, line width = 1.0] (-20.35,3.26) -- (-19.65,3.26);
	
	% create the dynamics
	
	\node at (-22.70, 4.47) {$\displaystyle \dot{x}_1 = -\sum_{j \in \mathcal{N}(1)}^{} (x_1 - x_j)+v_{1}$};
	\node at (-22.70, 3.95) {$\vdots$};
	\node at (-22.70, 3.26) {$\displaystyle \dot{x}_N = -\sum_{j \in \mathcal{N}(N)}^{} (x_N - x_j)+v_{N}$};
	
	% create control inputs v_1,...,v_N
	\draw [black, line width = 0.030cm] (-8.35, 4.80) -- (-8.35, 8.50);
	\draw [black, line width = 0.030cm] (-8.35, 8.50) -- (-18.00, 8.50);
	\draw [black, line width = 0.030cm] (-19.00, 8.50) -- (-23.00, 8.50);
	\draw[-latex, draw=black, line width = 1.0] (-23.00, 8.50) -- (-23.00, 5.1);
	
	\draw [black, line width = 0.030cm] (-8.35, 2.90) -- (-8.35, -0.80);
	\draw [black, line width = 0.030cm] (-8.35, -0.80) -- (-18.00, -0.80);
	\draw [black, line width = 0.030cm] (-19.00, -0.8) -- (-23.00, -0.8);
	\draw[-latex, draw=black, line width = 1.0] (-23.00, -0.8) -- (-23.00, 2.6);
	
	\draw[green!70, line width=.04cm] (-19.0, 7.95) rectangle +(1.0, 1.0);
	\node at (-18.5, 8.40) {$v_1(t)$};
	\draw[green!70, line width=.04cm] (-19.0, -1.25) rectangle +(1.0, 1.0);
	\node at (-18.5, -0.80) {$v_N(t)$};
	\end{tikzpicture}
	\centering
	\caption{A graphic illustration of the proposed framework.}
	\label{fig:solution_scheme}
\end{figure*}

In order to construct the Buchi WTSs $\widetilde{\mathcal{T}}_p$ and $\widetilde{\mathcal{T}}_i, i \in \mathcal{I}$ that were presented in Steps 2 and 4, consider the following generic definition:
\begin{definition} \label{def: buchi_WTS}
	Given a WTS $\mathcal{T}_i =(S_i, S_{i}^{\text{init}}, Act_i, \longrightarrow_i, d_i, AP_i, L_i)$, and a TBA $\A_i = (Q_i,  Q^\text{init}_i, C_i, Inv_i, E_i, F_i, \\ AP_i, \mathcal{L}_i)$ with $|C_i|$ clocks and let $C^{\mathit{max}}_i$ be the largest constant appearing in $\A_i$. Then, their \textit{B\"uchi WTS} $\widetilde{\T}_i = \mathcal{T}_i \otimes \A_i = (\widetilde{S}_i, \widetilde{S}_{i}^{\init}, \widetilde{Act}_i, {\rightsquigarrow}_{i}, \widetilde{d}_i, \widetilde{F}_i, AP_i, \widetilde{L}_i)$ is defined as follows:
	\begin{itemize}
		\item {$\widetilde{S}_i \subseteq \{(s_i, q_i) \in S_i \times Q_i : {L}_i(s_i) = \mathcal{L}_i(q_i)\} \times \mathbb{T}_\infty^{|C_i|} $.}
		\item $\widetilde{S}_{i}^{\init} = S_i^{\init} \times Q_i^{\init} \times \{0\}^{|C_i|}$.
		\item $\widetilde{Act}_i = Act_i$.
		\item $(\widetilde{q}, act_i, \widetilde{q} ') \in {\rightsquigarrow}_i$ iff
		\begin{itemize}
			\item[$\circ$] $\widetilde{q} = (s, q, \nu_1, \ldots, \nu_{|C_i|}) \in \widetilde{S}_i$, \\ $\widetilde{q} ' = (s', q', \nu_1', \ldots, \nu_{|C_i|}') \in \widetilde{S}_i$,
			\item[$\circ$] $act_i \in Act_i$,
			\item[$\circ$] $(s, act_i, s') \in \longrightarrow_i$, and
			\item[$\circ$] there exists $\gamma, R$, such that $(q, \gamma, R, q') \in E_i$, $\nu_1,\ldots,\nu_{|C_i|} \models \gamma$, $\nu_1',\ldots,\nu_{|C_i|}' \models Inv_i(q')$, and for all $i \in \{1,\ldots, |C_i|\}$ it holds that:
			\begin{equation*}
			\nu_i' =
			\begin{cases}
			0,      & \text{if } c_i \in R \\
			\nu_i + d_i(s, s'), &  \text{if }  c_i \not \in R \text{ and } \\ & \nu_i + d_i(s, s') \leq C^{\mathit{max}}_i \\			\infty, & \text{otherwise}.
			\end{cases}
			\end{equation*}
		\end{itemize}
		Then, $\widetilde{d}_i(\widetilde{q}, \widetilde{q}') = d_i(s, s')$.
		\item $\widetilde{F}_i = \{(s_i, q_i,\nu_1,\ldots,\nu_{|C_i|}) \in Q_i : q_i \in F_i\}$.
		\item $\widetilde{L}_i(s_i, q_i, \nu_1, \ldots, \nu_{|C_i|}) = {L}_i(s_i)$.
	\end{itemize}
\end{definition}

The Buchi WTS $\widetilde{\mathcal{T}}_p$ is constructed in a similar way to Def. \ref{def: buchi_WTS} by using the product of $\mathcal{T}_p$ and $\mathcal{A}_p$. Each B\"uchi WTS $\widetilde{\mathcal{T}}_i, i \in \mathcal I$ is in fact a WTS with a B\"uchi acceptance condition $\widetilde{F}_i$. A timed run of $\widetilde{\mathcal{T}}_i$ can be written as $\widetilde{r}_i^t = (q_i(0), \tau_i(0))(q_i(1), \tau_i(1)) \ldots$ using the terminology of Def. \ref{run_of_WTS}. It is \textit{accepting} if $q_i(j) \in \widetilde F_i$ for infinitely many $i \geq 0$.
An accepting timed run of  $\widetilde{\mathcal{T}}_i$ projects onto a timed run of $\mathcal{T}_i$ that satisfies the local specification formula $\varphi_i$ by construction. Formally, the following lemma, whose proof follows directly from the construction and and the principles of automata-based LTL model checking (see, e.g., \cite{katoen}), holds:

\begin{lemma} \label{eq: lemma_1}
	Consider an accepting timed run $\widetilde{r}_i^t = (q_k(0), \tau_i(0))(q_i(1), \tau_i(1)) \ldots$ of the B\"uchi WTS $\widetilde{\mathcal{T}}_i$ defined above, where $q_i(j) = (r_i(j), s_i(j), \nu_{i, 1}, \ldots, \nu_{i, |C_i|})$ denotes a state of $\mathcal{\widetilde T}_i$, for all $j \geq 0$. The timed run $\widetilde{r}_i^t$ projects onto the timed run $r_i^t = (r_i(0), \tau_i(0))(r_i(1), \tau_i(1)) \ldots$, of the WTS $\mathcal{T}_i$ that produces the timed word $w(r_i^t) = (L_i(r_i(0))$, $\tau_i(0))(L_i(r_i(1))$, $\tau_i(1)) \ldots$ accepted by the TBA $\mathcal{A}_i$ via its run $\rho_i = s_i(0)s_i(1) \ldots$. Vice versa, if there exists a timed run $r_i^t = (r_i(0),\tau_i(0))(r_i(1),\tau_i(1))\ldots,$ of the WTS $\mathcal{T}_k$ that produces a timed word $w(r_i^t) = (L_i(r_i(0))$, $\tau_i(0))(L_i(r_i(1))$, $\tau_i(1)) \ldots$ accepted by the TBA $\mathcal{A}_i$ via its run $\rho_i = s_i(0)s_i(1)\ldots$ then there exists the accepting timed run $\widetilde{r}_i^t = (q_i(0)$, $\tau_i(0))(q_i(1)$, $\tau_i(1)) \ldots$ of $\widetilde{\mathcal{T}}_i$, where $q_i(j) = (r_i(j),s_i(j),\nu_{i,1}, \ldots, \nu_{i,|C_i|})$, in $\widetilde{\mathcal{T}}_i$.
\end{lemma}

The proposed framework is depicted in Figure \ref{fig:solution_scheme}. The dynamics \eqref{eq: system} of each agent $i$ is abstracted into a WTS $\mathcal{T}_i$ (orange rectangles). Then the product between each WTS $\mathcal{T}_i$ and the $TBA$ $\mathcal{A}_i$ is computed according to Def. \ref{def: buchi_WTS}. The TBA $\mathcal{A}_i$ accepts all the words that satisfy the formula $\varphi_i$ (blue rectangles). For every B\"uchi WTS $\widetilde{\mathcal{T}}_i$ the controller synthesis procedure that was described in this Section (red rectangles) is performed and a sequence of accepted runs $\{\widetilde{r}_1^t, \dots, \widetilde{r}_N^t\}$ is designed. Every accepted run $\widetilde{r}_i^t$ maps into a decentralized controller $v_i(t)$ which is a solution to Problem \ref{problem: basic_prob}.

\begin{proposition}
	A solution obtained from Steps 1-5, gives a sequence of controllers $v_1, \dots, v_N$ that guarantees the satisfaction of the formulas $\varphi_1, \dots, \varphi_N$ of the agents $1, \dots, N$ respectively, governed by the dynamics as in \eqref{eq: system}, thus, they are a solution to Problem \ref{problem: basic_prob}. 	
\end{proposition}

\subsection{Complexity} \label{sec:complexity}

Denote by $|\varphi|$ the length of an MITL formula $\varphi$. A TBA $\mathcal{A}_i, i \in \mathcal{I}$ can be constructed in space and time $2^{\mathcal{O}(|\varphi_i|)}, i \in \mathcal{I}$. Let $\varphi_{\text{max}} = \underset{i \in \mathcal{I}}{\text{max}} \{ |\varphi_i|\}$ be the MITL formula with the longest length. Then, the complexity of Step 1 is $N 2^{\mathcal{O}(|\varphi_{\text{max}})|}$. Step 2 costs $\mathcal{O}(N 2^{|\varphi_i|}|\mathcal{S}_i|)$, where $|\mathcal{S}_i| = |\hat{\mathbb{I}}|$ is the number of states of the WTS $\mathcal{T}_i$. We have the best case complexity as $\mathcal{O}(N R_{\text{selec}} 2^{|\varphi_{\text{max}}|} |\hat{\mathbb{I}}|)$, since the Step 3 is more efficient than Step 4. The worst case complexity of our proposed framework is when Step 4 is followed, which is $\mathcal{O}(2^{|\varphi_{\text{max}}|}|\hat{\mathbb{I}}|^N)$.

\clearpage 

\begin{figure}[t!]
	\centering
	\begin{subfigure}[b]{0.45\linewidth}
		\centering
		\includegraphics[scale = 0.5]{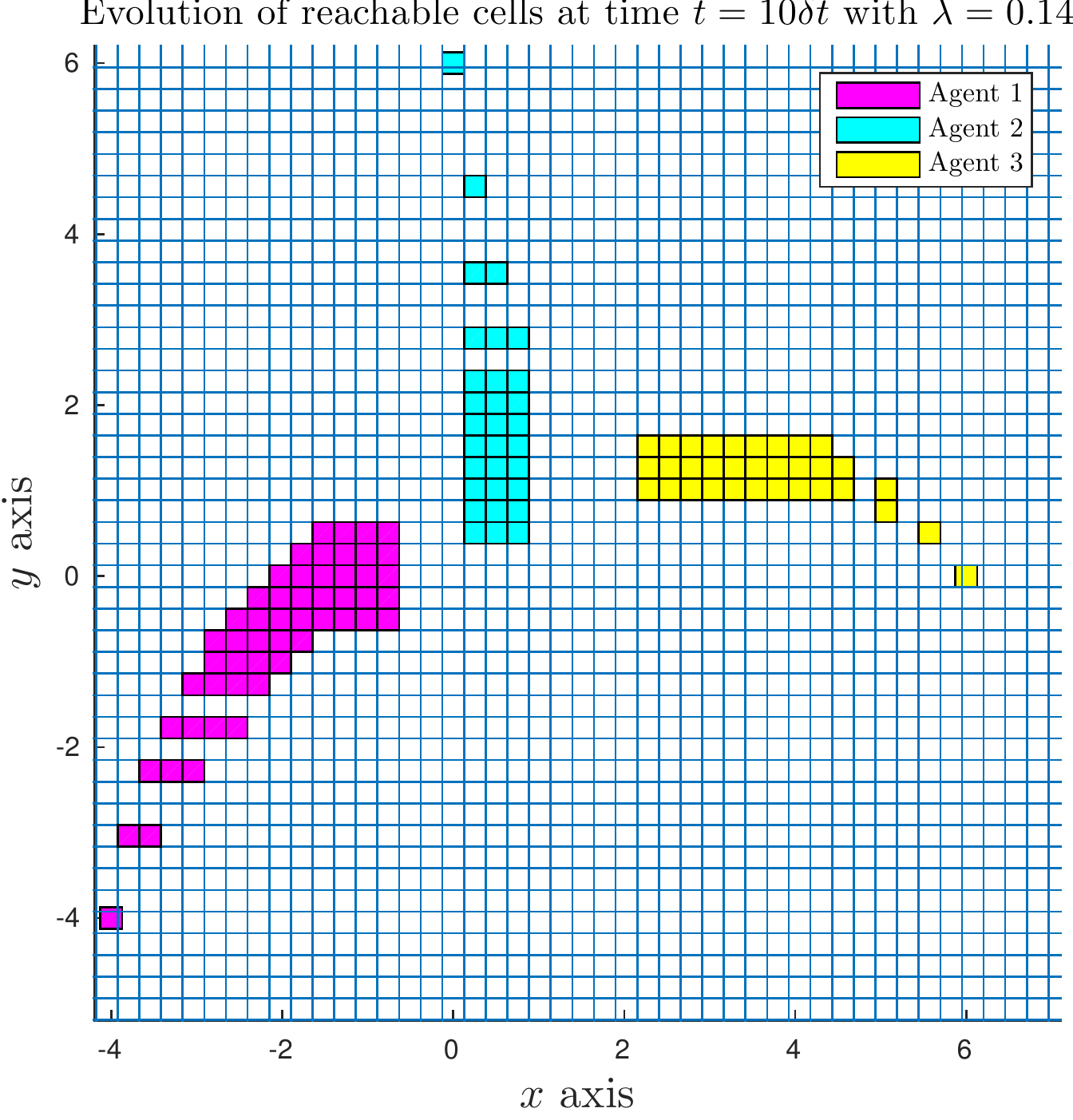}
		\caption{Scenario $1$ with $\lambda = 0.14$}\label{fig:1a}		
	\end{subfigure}
	\quad
	\begin{subfigure}[b]{0.45\linewidth}
		\centering
		\includegraphics[scale = 0.5]{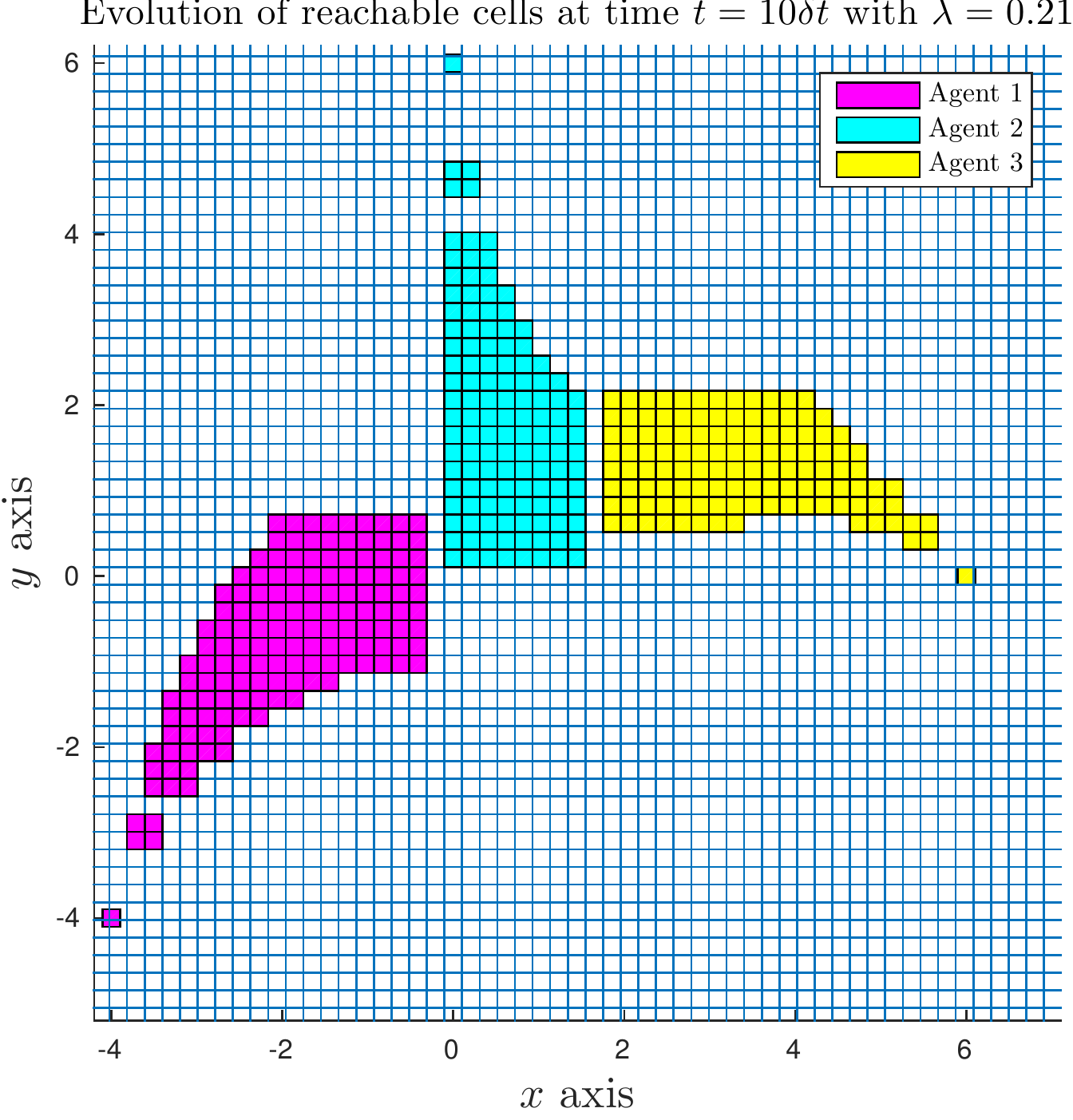}
		\caption{Scenario $2$ with $\lambda = 0.21$}\label{fig:1b}
	\end{subfigure}
	\caption{Two simulation scenarios with $N = 3$ agents and $\lambda = 0.14$, $\bar{d}_{\max} = 0.25$ (Figure \ref{fig:1a}) and $\lambda = 0.21$, $\bar{d}_{\max} = 0.20$ (Figure \ref{fig:1b})) in a time horizon of $t = 10 \delta t$.}\label{fig:simulation1}
\end{figure}

	\begin{table}[t!]
		\begin{center}
			%\begin{tabular}{|m{0.5cm}|m{2.2cm}|m{3.5cm}|m{0.6cm}|}
			\begin{tabular}{|C{0.9cm}||C{2.3cm}|C{1.6cm}||C{2.3cm}|C{1.6cm}|}
					\hline
					%\multicolumn{3}{}{Country List} & {asdf} \\
					\multicolumn{5}{|c|}{$N = 3$ agents}\\
					\hline
					\multicolumn{3}{|c||}{$\lambda = 0.14$} & \multicolumn{2}{c|}{$\lambda = 0.21$}\\
					\hline \hline
					Step  & Reachable States & Time &  Reachable States & Time  \\
					\hline \hline
					$\delta t$  & $2$ & $0.13 \sec$  & $24$ & $0.15 \sec$\\
					\hline
					$2 \delta t$  & $12$ & $0.05 \sec$ & $235$ & $0.11 \sec$\\
					\hline
					$3 \delta t$  & $18$ & $0.05 \sec$ & $1006$ & $0.88 \sec$\\
					\hline
					$4 \delta t$  & $54$ & $1.0 \sec$ & $3104$ & $5.17 \sec$\\
					\hline
					$5 \delta t$  & $250$ & $0.90 \sec$ &  $7983$ & $20.80 \sec$\\
					\hline
					$6 \delta t$  & $250$ & $1.27 \sec$ &  $12545$ & $48.77 \sec$\\
					\hline
					$7 \delta t$  & $250$ & $1.36 \sec$ & $13493$ & $70.21 \sec$\\
					\hline
					$8 \delta t$  & $475$ & $1.31 \sec$ & $16078$ & $74.56 \sec$\\
					\hline
					$9 \delta t$  & $875$ & $2.21 \sec$ & $23690$ & $107.79 \sec$\\
					\hline
					$10 \delta t$  & $1100$ & $3.09 \sec$ & $33171$ & $185.91 \sec$\\
					\hline \hline
					$ $ & \multicolumn{2}{|c||}{Total Time: $11.41 \sec$} & \multicolumn{2}{c|}{Total Time: $514.37 \sec$}\\
					\hline
				\end{tabular}
			\end{center}
			\caption{This table shows the simulation statistics of an example with $N = 3$ agents whose reachable cells are depicted in Figure \ref{fig:simulation1}. The first column is the number of the time step (recall that the duration of each transition in the WTS is exactly $\delta t$). The other columns show for $\lambda = 0.14$ and $\lambda = 0.21$, the number of the reachable states generated by the product WTS $\mathcal{T}_p$, as given in Definition $8$ of the revised version, as well as the required computational time. The total time stands for the computation of the reachable states.}
			\label{table1}
		\end{table}

\section{Simulation Results} \label{sec: simulation_results}

In order to show how the proposed framework scales with respect to the number of agents and the time solution horizon, we consider two simulation examples with two simulation scenarios each.

\textbf{Simulation Example 1 :}  Consider a system of three agents with $x_i \in \mathbb{R}^2, \ i \in \mathcal{I} = \{1, 2, 3\}, \mathcal{N}(1) = \{2\} = \mathcal{N}(3), \mathcal{N}(2) = \{1, 3\}$ is considered. According to \eqref{eq: system}, the dynamics are given as: $\dot{x}_1 = -(x_1-x_2)+v_1, \dot{x}_2 = -(x_2-x_1)-(x_2-x_3)+v_2$ and $\dot{x}_3 = -(x_3-x_2)+v_3$. The simulation parameters are set to $L_1 = \sqrt{2}$, $L_2 = 2$ and $\delta t = 0.1$. The initial agents' positions are set to $(-4,4), (0,6)$ and $(7,0)$ respectively. We consider Scenario 1 and Scenario 2 with $\lambda = 0.14$, $\bar{d}_{\max} = 0.25$ and $\lambda = 0.21$, $\bar{d}_{\max} = 0.20$, as is depicted in Figure \ref{fig:1a} and Figure \ref{fig:1b}, respectively. The cell decomposition presented in this paper and the reachable cells of each agent are depicted in Figure \ref{fig:1a}-\ref{fig:1b}. The reachable cells of each agent are depicted with purple, cyan and yellow respectively. In Figure \ref{fig:1a}-\ref{fig:1b} we can observe the evolution of the reachable sets of each agent at time $t = 10 \delta t$. It can be observed that the agents are not necessarily moving between neighboring cells and not all the individual runs satisfy the desired specification. The simulation statistics are depicted in Table \ref{table1}. The simulations were carried out in MATLAB Environment on a desktop with 8 cores, 3.60GHz CPU and 16GB of RAM.

\textbf{Simulation Example 2 :} Consider a multi-agent system with $x_i \in \mathbb{R}^2$, $i \in \mathcal{I} = \{1, 2, 3, 4\}$, $\mathcal{N}(1) = \{2\}, \mathcal{N}(2) = \{1,3\}$, $\mathcal{N}(3) = \{2, 4\}$, $\mathcal{N}(4) = \{3\}$. According to \eqref{eq: system}, the dynamics are given as: $\dot{x}_1 = -(x_1-x_2)+v_1, \dot{x}_2 = -(x_2-x_1)-(x_2-x_3)+v_2$, $\dot{x}_3 = -(x_3-x_2)-(x_2-x_4)+v_3$ and $\dot{x}_4 = -(x_4-x_3)+v_4$. The simulation parameters are set to $L_1 = \sqrt{2}$, $L_2 = 2$, $\delta t = 0.1$. The workspace is decomposed into square cells, which are depicted with blue color in Figure \ref{fig:simulation2}. The initial agents' positions are set to $(-4,4)$, $(0,6)$, $(7,0)$ and $(4,-5)$, respectively. We consider Scenario 1 and Scenario 2 with $\lambda = 0.14$, $\bar{d}_{\max} = 0.25$ and $\lambda = 0.21$, $\bar{d}_{\max} = 0.20$, as is depicted in Figure \ref{fig:1a} and Figure \ref{fig:1b}, respectively. The specification formulas for the Scenario $1$ are set to $\varphi_1 = \Diamond_{[0.2, 1.0]} \{\rm green\}$, $\varphi_2 = \Diamond_{[0.1, 1.0]} \{\rm orange\}$, $\varphi_3 = \Diamond_{[0.5, 1.5]} \{\rm black\}$ and $\varphi_4 = \Diamond_{[0.3, 2.3]} \{\rm gray\}$, respectively. The cell decomposition presented in this paper, the reachable cells of each agent up to time $t = 10 \delta t$ and the goal regions are depicted in Figure \ref{fig:simulation2}. The reachable cells of each agent are depicted with purple, cyan and yellow respectively. Note that the agents' transitions are not necessarily performed between neighboring cells. The individual consistent runs $\widetilde{r}_1^t$, $\widetilde{r}_2^t$, $\widetilde{r}_3^t$ and $\widetilde{r}_4^t$ of agents $1$, $2$, $3$ and $4$ that satisfy the formulas $\varphi_1$, $\varphi_2$, $\varphi_3$ and $\varphi_4$, respectively are depicted in Figure \ref{fig:simulation2} with black arrows. Each arrow represents a transition from a state to another according to Def. \ref{def: indiv_WTS}. The product WTS $\mathcal{T}_p$ has $52877$ reachable states in case of $\lambda = 0.14$ and $1255547$ when $\lambda = 0.21$ since, according to Section 4.2.2, larger values of $\lambda$ lead to finer discretization. Agent $1$ satisfies $\varphi_1$ in $5 \delta t$, agent $2$ satisfies $\varphi_2$ in $2 \delta t$, agent $3$ satisfies $\varphi_3$ in $8 \delta t$ and agent $4$ satisfies $\varphi_4$ in $10 \delta t$. The simulation is performed in a horizon of $10$ steps and it takes $668.80 \sec$ ($609.5 \sec$ for the abstraction and $59.30 \sec$ for the graph search) for Scenario 1.

\begin{figure}[t!]
	\centering
	\begin{subfigure}[b]{0.45\linewidth}
		\centering
		\includegraphics[scale = 0.5]{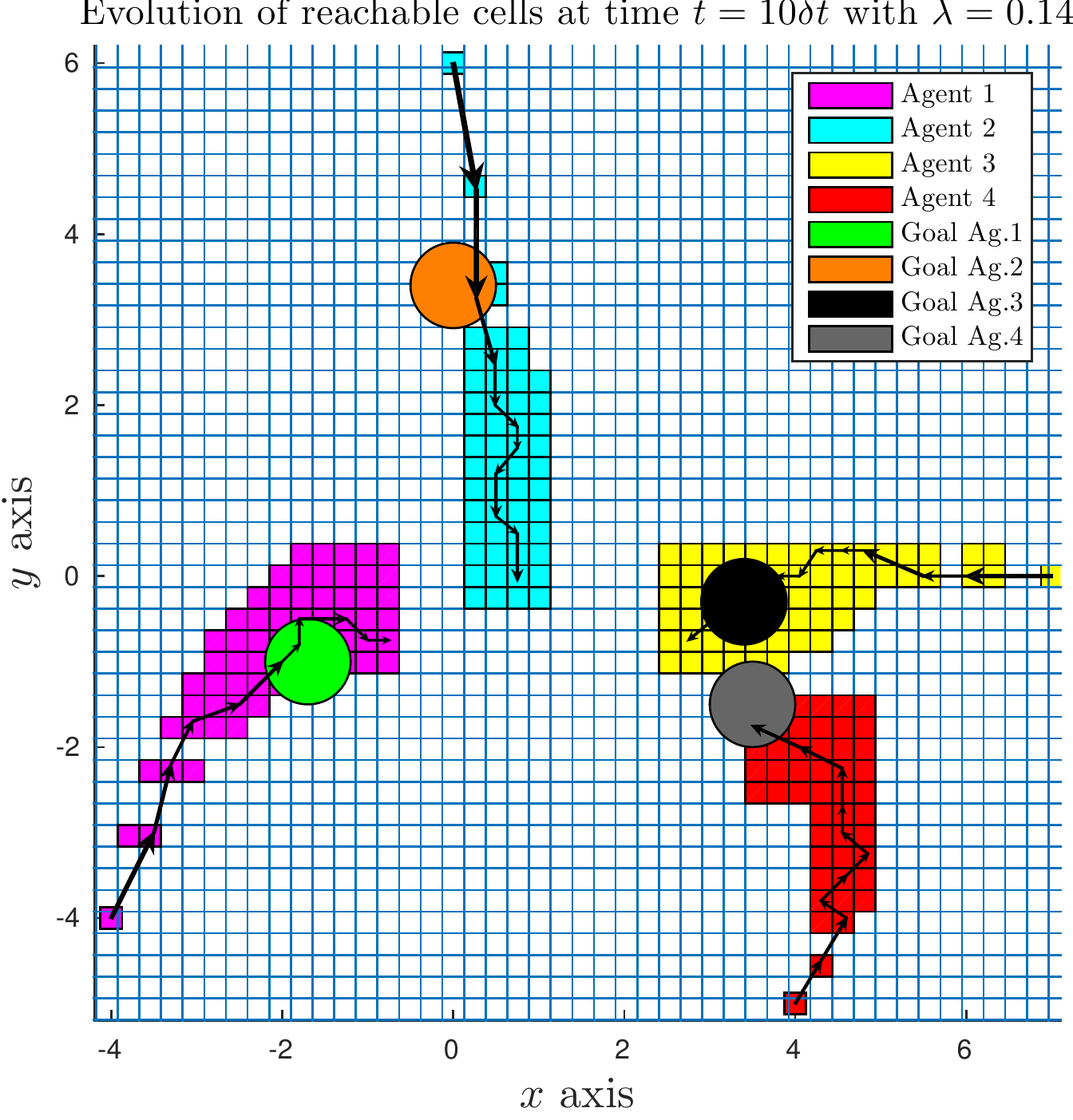}
		\caption{Scenario $1$ with $\lambda = 0.14$}\label{fig:2a}		
	\end{subfigure}
	\quad
	\begin{subfigure}[b]{0.45\linewidth}
		\centering
		\includegraphics[scale = 0.5]{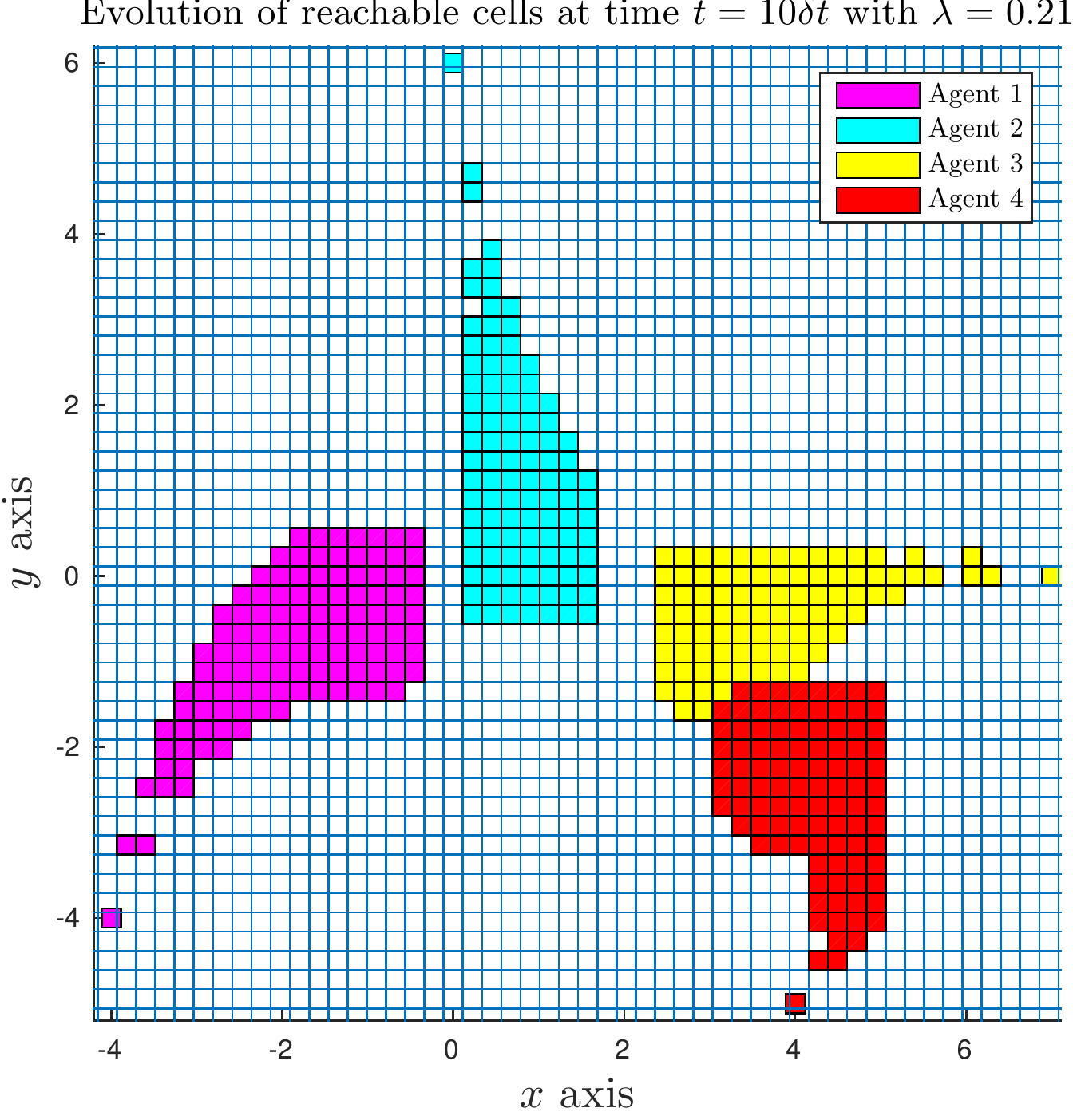}
		\caption{Scenario $2$ with $\lambda = 0.21$}\label{fig:2b}
	\end{subfigure}
	\caption{Two simulation scenarios with $N = 3$ agents and $\lambda = 0.14$, $\bar{d}_{\max} = 0.25$ (Figure \ref{fig:2a}) and $\lambda = 0.21$, $\bar{d}_{\max} = 0.20$ (Figure \ref{fig:2b})) in a time horizon of $t = 10 \delta t$. The figure depicts the evolution of the agents' reachable cells up to time $t = 10 \delta t$. The scenario is the same with the simulation of Section 6 of the revised paper, but now we have removed agent $4$. The rest of agents start from the same initial conditions as in the previous Scenarios.}\label{fig:simulation2}
\end{figure}

\begin{table}[t!]
	\begin{center}
		%\begin{tabular}{|m{0.5cm}|m{2.2cm}|m{3.5cm}|m{0.6cm}|}
		\begin{tabular}{|C{0.9cm}||C{2.3cm}|C{1.6cm}||C{2.3cm}|C{2.1cm}|}
			\hline
			%\multicolumn{3}{}{Country List} & {asdf} \\
			\multicolumn{5}{|c|}{$N = 4$ agents}\\
			\hline
			\multicolumn{3}{|c||}{$\lambda = 0.14$} & \multicolumn{2}{c|}{$\lambda = 0.21$}\\
			\hline \hline
			Step  & Reachable States & Time &  Reachable States & Time  \\
			\hline \hline
			$\delta t$  & $8$ & $0.22 \sec$  & $24$ & $0.32 \sec$\\
			\hline
			$2 \delta t$  & $39$ & $0.11 \sec$ & $208$ & $0.24 \sec$\\
			\hline
			$3 \delta t$  & $216$ & $0.16 \sec$ & $6702$ & $5.92 \sec$\\
			\hline
			$4 \delta t$  & $1610$ & $17.97 \sec$ & $26843$ & $74.95 \sec$\\
			\hline
			$5 \delta t$  & $3168$ & $16.73 \sec$ &  $89817$ & $91.68 \sec$\\
			\hline
			$6 \delta t$  & $5346$ & $21.87 \sec$ &  $133904$ & $252.12 \sec$\\
			\hline
			$7 \delta t$  & $5808$ & $31.55 \sec$ & $222037$ & $451.61 \sec$\\
			\hline
			$8 \delta t$  & $10168$ & $37.39 \sec$ & $358941$ & $1427.65 \sec$\\
			\hline
			$9 \delta t$  & $23004$ & $93.98 \sec$ & $644489$ & $4803.12 \sec$\\
			\hline
			$10 \delta t$  & $52877$ & $389.50 \sec$ & $1255547$ & $19321.07 \sec$\\
			\hline \hline
			$ $ & \multicolumn{2}{|c||}{Total Time: $609.50 \sec$} & \multicolumn{2}{c|}{Total Time: $26429 \sec$}\\
			\hline
		\end{tabular}
		\end{center}
		\caption{This table shows the simulation statistics of an example with $N =4$ agents whose reachable cells are depicted in Figure \ref{fig:simulation2}. The number of reachable states of each time step is depicted. The total simulation time stands for the abstraction procedure.}
		\label{table2}
	\end{table}

\section{Conclusions and Future Work} \label{sec: conclusions}

A systematic method for controller synthesis of dynamically coupled multi-agent path-planning has been proposed, in which timed constraints of fulfilling a high-level specification are imposed to the system. The solution involves a boundedness analysis, the abstraction of each agent's motion into WTSs, TBAs as well as B\"uchi WTSs construction. The simulation example demonstrates our solution approach. Future work includes further computational improvement of the abstraction method and more complicated high-level tasks being imposed to the agents in order to exploit the expressiveness of MITL formulas.

\bibliographystyle{unsrt}        % Include this if you use bibtex 
\bibliography{references}           % and a bib file to produce the 

\appendix

\section{Proof of Theorem 1} \label{app:proof_theorem_1}

Consider the following candidate Lyapunov function $V: \mathbb{R}^{Nn} \rightarrow \mathbb{R}_{\ge 0}$
\begin{equation*}
V(x) = \frac{1}{2} \sum_{i=1}^{N} \sum_{j \in \mathcal{N}(i)}^{} \| x_i-x_j\|^2 = \|\widetilde{x}\|^2 > 0.
\end{equation*}
The time derivative of $V$ along the trajectories of \eqref{eq: system}, can be computed as
\begin{align}
&\hspace{-1mm}\dot{V} = \left[ \nabla V(x) \right]^{\top} \ \dot{x} \nonumber \\
&\hspace{-1mm}= \sum_{k=1}^{n} \left\{ c\left(\frac{\partial V}{\partial x} ,k\right)^{\top} \left[ -L(\mathcal{G}) \ c(x,k)+c(v, k) \right] \right\}, \label{lyap1}
\end{align}
where $\displaystyle c\left(\frac{\partial V}{\partial x} ,k \right)  = \begin{bmatrix} \frac{\partial V}{\partial x_1^k} & \dots & \frac{\partial V}{\partial x_N^k} \end{bmatrix}^{\top}$. By computing the partial derivative of the Lyapunov function with respect to vector $x_i, \ i \in \mathcal{I}$ we get $\displaystyle \frac{\partial V}{\partial x_i} = \sum_{j \in \mathcal{N}(i)}^{} (x_i-x_j), \ i \in \mathcal{I}$ from which we have that $\displaystyle c\left(\frac{\partial V}{\partial x} ,k\right)^{\top} = c(x,k)^{\top} \ L(\mathcal{G}), \ k = 1,...,n$. Thus, by substituting the last in \eqref{lyap1} we get
\begin{align}
\dot{V} &= -\sum_{k=1}^{n} \left\{ c(x,k)^{\top} \ \left[ L(\mathcal{G})\right]^2 \ c(x,k) \right\} + \sum_{k=1}^{n} \left\{ c(x,k)^{\top} \ \left[ L(\mathcal{G})\right]^2 c(v, k) \right\} \nonumber \\
&\leq -\sum_{k=1}^{n} \left\{ c(x,k)^{\top} \ \left[ L(\mathcal{G})\right]^2 \ c(x,k)) \right\} + \left\| \sum_{k=1}^{n} \left\{ c(x,k)^{\top} \ L(\mathcal{G}) \ c(v, k) \right\} \right\|. \label{eq: lyap3}
\end{align}

\noindent For the first term of \eqref{eq: lyap3} we have that
\begin{align}
\sum_{k=1}^{n} \left\{ c(x,k)^{\top} \ L(\mathcal{G})^2 \ c(x,k)) \right\} = \sum_{k=1}^{n}  \left\| L(\mathcal{G}) \ c(x,k) \right\|^2. \notag
\end{align}

\noindent For the second term of \eqref{eq: lyap3} we get
\begin{align}
\left\| \sum_{k=1}^{n} \left\{ c(x,k)^{\top} \ L(\mathcal{G}) \ c(\nu,k) \right\} \right\| & = \left\| \sum_{k=1}^{n} \left\{ c(x,k)^{\top} \ D(\mathcal{G}) \ D(\mathcal{G})^{\top} \ c(\nu,k) \right\} \right\|  \notag \\
& \leq \sum_{k=1}^{n} \left\{ \left\| D(\mathcal{G}) ^{\top} \ c(x,k) \right\| \ \left\| D(\mathcal{G})^{\top} \right\| \ \left\| c(\nu,k) \right\|\right\} \notag \\
& =  \left\| D(\mathcal{G})^{\top} \right\| \ \sum_{k=1}^{n} \left\{ \left\| c(\widetilde{x},k) \right\| \ \left\| c(\nu,k) \right\| \right\}. \label{eq: lyap4}
\end{align}
By using the Cauchy-Schwarz inequality in \eqref{eq: lyap4} we get	
\begin{align}
\left\| \sum_{k=1}^{n} \left\{ c(x,k)^{\top} \ L(\mathcal{G}) \ c(v, k) \right\} \right\| & \leq \left\| D(\mathcal{G})^{\top} \right\| \ \left( \sum_{k=1}^{n} \left\| c(\widetilde{x},k) \right\|^2 \right)^{\frac{1}{2}} \ \left( \sum_{k=1}^{n} \left\| c(v, k) \right\|^2 \right)^{\frac{1}{2}} \notag \\
& = \left\| D(\mathcal{G})^{\top} \right\| \|\widetilde{x}\| \|v\| \leq \left\| D(\mathcal{G})^{\top} \right\| \|\widetilde{x}\| \sqrt{N} \|v\|_\infty, \notag
\end{align}

\noindent where $\|v\|_\infty = \text{max} \left\{\|v_i\| : i \in \mathcal{I} \right\} \leq v_{\text{max}}$. Thus, by combining the previous inequalities, \eqref{eq: lyap3} is written
\begin{equation}
\dot{V} \leq -\sum_{k=1}^{n}  \left\{ \left\| L(\mathcal{G}) c(x,k) \right\|^2 \right\} + \sqrt{N} \left\| D(\mathcal{G})^{\top} \right\| \|\widetilde{x}\| v_{\text{max}}. \label{eq: lyap8}
\end{equation}
By exploiting Lemma \ref{lemma: lemma_1}, \eqref{eq: lyap8} is written as:
\begin{align}
\dot{V} &\leq -\lambda_2^2(\mathcal{G}) \sum_{k=1}^{n} \left\{\left\| c(x^{\perp},k) \right\|^2 \right\} + \sqrt{N} \left\| D(\mathcal{G})^{\top} \right\| \|\widetilde{x}\| v_{\text{max}} \notag \\
&= -\lambda_2^2(\mathcal{G}) \ \|x^\perp\|^2  + \sqrt{N} \left\| D(\mathcal{G})^{\top} \right\| \|\widetilde{x}\| v_{\text{max}} \notag \\
&\leq -\frac{\lambda_2^2(\mathcal{G})}{2(N-1)}\|\widetilde{x}\|^2  + \sqrt{N} \left\| D(\mathcal{G})^{\top} \right\| \|\widetilde{x}\| v_{\text{max}} \notag \\
&\leq -K_1 \|\widetilde{x}\| \left(\|\widetilde{x}\|-K_2 v_{\text{max}} \right). \label{eq:concl_1}
\end{align}	
where $K_1 = \frac{\lambda_2^2(\mathcal{G})}{2(N-1)} > 0$. By using the following implication $\widetilde{x} = D^{\top}(\mathcal{G})x \Rightarrow \|\widetilde{x}\| = \|D(\mathcal{G})^\top x\| \leq \|D(\mathcal{G})^\top\| \|x\|$, apparently, we have that $0 < V(x)  = \|\widetilde{x}\|^2 \leq \|D(\mathcal{G})^\top\|^2 \|x\|^2 \ \text{and} \ \dot{V}(x) < 0$ when $\|\widetilde{x}\| \geq \bar{R} > K_2 v_{\text{max}} $. Thus, there exists a finite time $T > 0$ such that the trajectory will enter the compact set $\mathcal{X} = \{x \in \mathbb{R}^{Nn} : \|\widetilde{x}\| \leq \bar{R}\}$ and remain there for all $t \geq T$ with $\bar{R} > K_2 v_{\text{max}}$. This can be extracted from the following. Let us define the compact set $$\Omega = \left\{ x \in \mathbb{R}^{Nn} : K_2 v_{\text{max}} < \bar{R} \leq \|\widetilde{x}\| \leq \bar{M} \right\},$$ where $\bar{M} = V(x(0)) = \|\widetilde{x}(0)\|^2$. Without loss of generality it is assumed that it holds $\bar{M} > \bar{R}$. Let us define the compact sets $\mathcal{S}_{1} = \left\{ x \in \mathbb{R}^{Nn} : \|\widetilde{x}\| \leq \bar{M} \right\}, \mathcal{S}_{2} = \left\{ x \in \mathbb{R}^{Nn} : \|\widetilde{x}\| \leq K_2 v_{\text{max}} \right\}$. From the equivalences $\forall \ x \in S_1 \Leftrightarrow V(x) = \|\widetilde{x}\|^2 \leq \bar{M}^2, \forall \ x \in S_2 \Leftrightarrow V(x) = \|x\|^2 \leq K^2_2 v^2_{\text{max}}$, we have that the boundaries $\partial S_1, \partial S_2$ of sets $S_1, S_2$ respectively, are two level sets of the Lyapunov function $V$. By taking the above into consideration we have that $\partial S_2 \subsetneq \partial S_1$. Hence, we get from \eqref{eq:concl_1} that: 
\begin{equation}
\dot{V}(x) < 0, \forall \ x \in \Omega = S_1 \backslash S_2, \label{eq:concl_2}
\end{equation}			
In view of \eqref{eq:concl_2} and the fact that the sets $\mathcal{S}_1, \mathcal{S}_2$ are defined in terms of level sets of $V$, we conclude that both $S_1$ and $S_2$ are invariant with respect to the system \eqref{eq: system}. Consequently, according to \cite[Lemma 5.1]{liberzon_switching} the trajectory that starts inside the set $S_1$ has to enter the interior of the set of $S_2$ in finite time $T > 0$ and remain there for all time $t \geq T$.

\section{Proof of Lemma 2}

\begin{lemma} \label{lemma: lemma_1}
	Let $x^\perp$ be the projection of the vector $x \in \mathbb{R}^{Nn}$ to the orthogonal complement of the subspace $H = \{x \in \mathbb{R}^{Nn}: x_1 = \ldots = x_N\}$. Then, the following properties hold:
	\begin{align*}
	\| L(\mathcal{G}) \ c(x, k) \| &\geq \lambda_2 (\mathcal{G}) \ \|c(x^{\perp}, k) \|, \ \forall \ k \in \mathcal{I}, \\
	\|x^\perp\| &\geq \frac{1}{\sqrt{2(N-1)}} \|\widetilde{x}\|.
	\end{align*}
\end{lemma}

\begin{proof}
	The proof can be found in \cite[Appendix A]{boskos_cdc_connectivity}.
\end{proof}

\section{Sufficient Conditions for Well-Possessedness of the Abstraction} \label{app:suff_conditions}

We investigate here if the system \eqref{eq: system} satisfies the sufficient conditions \textbf{(C1)}-\textbf{(C3)} for well-posed abstractions.

\noindent \textbf{(C1)} For every $i \in \mathcal{I}, \forall \ x \in \mathbb{R}^{Nn} : \widetilde{x} \in \mathcal{X}$ and $\text{pr}_i(x) = (x_i, \mathbf{x}_j)$ it holds that:
\begin{align*}
\| f_i(x_i, \mathbf{x}_j) \| &= \left\| -\sum_{j \in \mathcal{N}(i)}^{}(x_i-x_j) \right\|  \leq \sum_{j \in \mathcal{N}(i)}^{} \| x_i-x_j \| \notag \\ 
&\leq  \sum_{ (i,j) \in \mathcal{E} }^{} \| x_i-x_j \| = \Delta x \leq \bar{R}.
\end{align*}
Thus, $M = \bar{R}$. We have also that $\|D(\mathcal{G})^\top\| = \sqrt{\lambda_{\text{max}} ( D(\mathcal{G}) D(\mathcal{G})^\top)} = \sqrt{\lambda_{\text{max}}(\mathcal{G})}$ and $\lambda_2(\mathcal{G}) \leq \frac{N}{N-1} \min \{N_i : i \in \mathcal{I} \}$ from \cite{fiedler1973algebraic}.  For $N > 2$ it holds that $\lambda_2(\mathcal{G}) < N$. From Theorem \eqref{theorem: theorem_1} we have that $\bar{R} > K_2 v_{\text{max}} \Leftrightarrow M > K_2 v_{\text{max}}$. It holds that $M > v_{\text{max}}$ since

\begin{align}
K_2 &= \frac{2 \sqrt{N} (N-1) \left\| D(\mathcal{G})^{\top} \right\|}{\lambda_2^2(\mathcal{G})} = \frac{2 \sqrt{N} (N-1) \sqrt{\lambda_{\text{max}}(\mathcal{G})}}{\sqrt{\lambda_2^3(\mathcal{G})} \sqrt{\lambda_2(\mathcal{G})}} \notag \\
&\geq \frac{2 \sqrt{N} (N-1)}{\sqrt{N^3}} \sqrt{\frac{\lambda_{\text{max}}(\mathcal{G})}{\lambda_2(\mathcal{G})}} \geq \frac{2 \sqrt{N} (N-1)}{\sqrt{N^3}} > 1. \notag
\end{align}

\noindent \textbf{(C2)} Starting from the left hand side of \eqref{eq:lipsitch_1} we get:
\begin{align}
\| f_i(x_i, \mathbf{x}_j) - f_i(x_i, \mathbf{y}_j) \| & = \Big\| -\sum_{j \in \mathcal{N}(i)}^{}(x_i-x_j) + \sum_{j \in \mathcal{N}(i)}^{}(x_i-y_j) \Big\| \notag \\
& \leq  \text{max}\{\sqrt{N_i} : i \in \mathcal{I}\} \ \|(x_i, \mathbf{x}_j) - (x_i, \mathbf{y}_j) \|. \notag
\end{align}
Thus, the condition \textbf{(C2)} holds and the Lipschitz constant is $L_1 = \text{max}\{\sqrt{N_i} : i \in \mathcal{I}\} >0$, where the inequality $\displaystyle \left( \sum_{i = 1}^{\rho} \alpha_i \right)^2 \leq \rho \ \left( \sum_{i=1}^\rho \alpha_i^2 \right)$ is used. \\

\noindent \textbf{(C3)} By using the same methodology as in \textbf{(C2)}, we conclude that $L_2 = \text{max} \{N_i : i \in \mathcal{I}\} > 0$.

\end{document}